\documentclass[11pt, one column]{article}
\usepackage{lmodern}
\usepackage{microtype}

\usepackage[letterpaper,margin=1in]{geometry}
\usepackage[T1]{fontenc}
\usepackage[utf8]{inputenc}
\usepackage[english]{babel}
\usepackage{microtype}
\usepackage{setspace}
\usepackage{lineno}          
\usepackage{amsmath}
\usepackage{hyperref}
\usepackage{amssymb}
\usepackage{amsthm}
\usepackage{mathtools}          
\usepackage{bm}                 
\usepackage[linesnumbered,ruled,vlined]{algorithm2e} 
\usepackage[mathscr]{eucal}     

\usepackage{aliascnt}           
\usepackage{thmtools}           

\usepackage{xcolor}             
\usepackage{soul}               
\usepackage{tcolorbox}
  \tcbuselibrary{skins}
\usepackage{mdframed}
\usepackage{framed}

\usepackage{float}              
\usepackage{environ}            
\usepackage{xspace}             
\usepackage{cancel}             
\usepackage{csquotes}           
\usepackage[colorinlistoftodos]{todonotes}
\usepackage{marginnote}
\usepackage{tikz}
\usepackage{placeins}
\usetikzlibrary{arrows.meta}

\usepackage{graphicx}
\usepackage{tabularx}
\usepackage{subcaption}
\usepackage{multirow}
\usepackage{array}
\usepackage{enumerate}
\usepackage{xstring}
\usepackage{comment}
\usepackage{listings}
\usepackage{cleveref}

\usepackage[maxbibnames=99, minbibnames=99]{biblatex}
\theoremstyle{plain}

\newtheorem{lemma}{Lemma}[section]
\newtheorem{corollary}[lemma]{Corollary}
\newtheorem{proposition}[lemma]{Proposition}

\newtheorem{claim}[lemma]{Claim}
\newtheorem*{claim*}{Claim}

\theoremstyle{definition}
\newtheorem{definition}[lemma]{Definition}

\newtheorem{observation}[lemma]{Observation}

\theoremstyle{remark}

\newtheorem*{remark*}{Remark}

\newtheorem*{note*}{Note}

\makeatletter
  \newcommand{\@problemtitle}{}
  \newcommand{\@probleminput}{}
  \newcommand{\@problemquestion}{}

  \NewEnviron{problem}{%
    \problemtitle{}\probleminput{}\problemquestion{}%
    \BODY
    \par\addvspace{.5\baselineskip}
    \noindent
    \begin{tabularx}{\textwidth}{@{\hspace{\parindent}} l X}
      \multicolumn{2}{@{\hspace{\parindent}}l}{\textbf{\@problemtitle}} \\
      \textbf{Input:} & \@probleminput \\
      \textbf{Goal:}  & \@problemquestion
    \end{tabularx}
    \par\addvspace{.5\baselineskip}
  }
\makeatother

\makeatletter
  \providecommand*{\cupdot}{%
    \mathbin{\mathpalette\@cupdot{}}%
  }
  \newcommand*{\@cupdot}[2]{%
    \ooalign{%
      $\m@th#1\cup$\cr
      \hidewidth$\m@th#1\cdot$\hidewidth
    }%
  }
\makeatother

\newcommand{\gs}{\mathsf{GS}}
\newcommand{\ws}{\mathsf{WS}}
\newcommand{\witN}{\mu_{WS}}

\newcommand{\vis}{\mathsf{Vis}}
\newcommand{\grN}{\mu_{GS}}

\newcommand{\bd}{\partial}
\newcommand{\inte}{\mathsf{iN}}
\newcommand{\ex}{\mathsf{eX}}

\newcommand{\wvp}{\mathsf{WV-Polygon}}
\newcommand{\WVP}{\mathscr{WVP}}
\newcommand{\SGP}{\mathsf{SGP}}
\newcommand{\AGP}{\mathsf{AGP}}
\newcommand{\ATG}{\mathsf{ATG}}

\newcommand{\RA}{\mathsf{rAnchor}}
\newcommand{\LA}{\mathsf{\ell Anchor}}

\newcommand{\I}{\mathtt{Int}}
\newcommand{\po}{\mathcal{P}}
\newcommand{\mo}{\mathcal{M}}

\newcommand{\OO}{\mathcal{O}}

\newcommand{\wv}{\mathcal{WV}}
\newcommand{\eb}{{e}_{\mathtt{base}}}

\newcommand{\OPT}{\mathrm{OPT}}
\newcommand{\spr}{\mathsf{\Pi_{v}}}
\newcommand{\spl}{\mathsf{\Pi_{u}}}
\newcommand{\SI}{\I_{\mathrm{strong}}}
\newcommand{\shpt}{\mathtt{shadowPoint}}

\newcommand{\np}{{\sf NP}\xspace}

\newcommand{\nph}{{\sf NP}-hard\xspace}
\newcommand{\npc}{{\sf NP}-complete\xspace}

\newcommand{\vsg}{\mathsf{VSG}}

\newcommand{\clv}[2][]{\ensuremath{#2_{#1}^{*}{_{\mathtt{clone}}}}}

\newcommand{\clvs}[2][]{\ensuremath{#2_{#1}{_{\mathtt{clone}}}}}

\crefname{theorem}{Theorem}{Theorems}       \Crefname{theorem}{Theorem}{Theorems}
\crefname{lemma}{Lemma}{Lemmas}             \Crefname{lemma}{Lemma}{Lemmas}
\crefname{claim}{Claim}{Claims}             \Crefname{claim}{Claim}{Claims}
\crefname{observation}{Observation}{Observations}
\Crefname{observation}{Observation}{Observations}
\crefname{definition}{Definition}{Definitions}
\Crefname{definition}{Definition}{Definitions}
\crefname{corollary}{Corollary}{Corollaries}
\Crefname{corollary}{Corollary}{Corollaries}
\crefname{proposition}{Proposition}{Propositions}
\Crefname{proposition}{Proposition}{Propositions}
\crefname{example}{Example}{Examples}       \Crefname{example}{Example}{Examples}
\crefname{note}{Note}{Notes}                \Crefname{note}{Note}{Notes}
\crefname{figure}{Figure}{Figures}         \Crefname{figure}{Figure}{Figures}
\crefname{subfigure}{Figure}{Figures}      \Crefname{subfigure}{Figure}{Figures}
\crefname{section}{Section}{Sections}      \Crefname{section}{Section}{Sections}
\crefname{equation}{Equation}{Equations}   \Crefname{equation}{Equation}{Equations}
\crefname{algocf}{Algorithm}{Algorithms}   \Crefname{algocf}{Algorithm}{Algorithms}
\crefformat{theorem}{Theorem~#2#1#3}        \Crefformat{theorem}{Theorem~#2#1#3}
\crefformat{lemma}{Lemma~#2#1#3}            \Crefformat{lemma}{Lemma~#2#1#3}
\crefformat{claim}{Claim~#2#1#3}            \Crefformat{claim}{Claim~#2#1#3}
\crefformat{observation}{Observation~#2#1#3}
\Crefformat{observation}{Observation~#2#1#3}

\newif\ifauthorcomments
\authorcommentstrue   

\newcounter{shouvik}

\newcounter{sasanka}

\newcounter{udvas}

\title{Perfectly Guarding Straits: Exact Algorithms for \\ Weak Visibility Polygons}
\author{%
  Shouvik Mondal\thanks{\texttt{shouvik.math@gmail.com}}
  \and
  Udvas Das\thanks{\texttt{udvas.das@gmail.com}}
  \and
  Sasanka Roy\thanks{\texttt{sasanka.ro@gmail.com}}
}
\date{Advanced Computing and Microelectronics Unit,\\
      Indian Statistical Institute, Kolkata, India}

\begin{document}

\maketitle
\bigskip

\begin{abstract}
The {\sc Art Gallery Problem} ($\AGP$) asks for the fewest guards that see all of a simple polygon. It is $\exists\mathbb{R}$-complete \cite{DBLP:journals/jacm/AbrahamsenAM22}, hence \nph{}. We show that for a \textit{particular class of polygons}, confining guards to a single edge makes $\AGP$\ exactly and efficiently solvable. We call this the \emph{{\sc Strait Guarding Problem}} ($\SGP$). Its input is a \emph{weak visibility polygon} ($\WVP$): a simple polygon where every point is seen from some point of one fixed edge, the \emph{base}. $\SGP$\ places the fewest guards on the base that jointly see the whole polygon. First, a structural fact: guards on $\eb$ that cover the boundary already cover the entire interior, turning a two-dimensional covering problem into a one-dimensional one. Our main result is the \emph{Witness-Guard Algorithm}, which solves $\SGP$\ exactly in $\OO\bigl((n + \OPT\cdot\rho)\,(\log n + \log\OPT)\bigr)$ time, where $\rho$ is the number of reflex vertices in the $\WVP$ and $\OPT$ is the minimum number of guards. It is output-sensitive and certifies optimality by a size-$\OPT$ witness set derived from its output. We also study the guarding-the-vertex version and prove a tight $\Theta(n\log n)$ bound, with the lower bound following from Sorting. As a corollary of $\SGP$, we obtain two results for altitude terrains ($\ATG$), a special case that $\SGP$ generalizes. We give a linear-time perfect-guarding algorithm, improving the previous $\OO(n^{2}\log n)$ bound \cite{DAESCU201922}. We also resolve the problem of \cite{DAESCU201922} on the minimum guarding altitude, in $\OO(nk+k^{2}\log k)$ time, which is an improvement on $\OO(k^{2}\lambda_{k-1}(n)\log n)$ bound \cite{DBLP:conf/iwoca/KangKA25}.
\end{abstract}


\section{Introduction}\label{sec:Intro}

Picture a strait, a narrow body of water with one coastline on each side. A country wants to monitor the entire strait, but it can only place guards, cameras, or radar posts along its own shore, not out in the water or on the far coast. The same picture shows up whenever a sensor is stuck on a rail, a camera is bolted to one wall, or a border patrol agent can only walk along one fence line. In every case, the region to be watched is large and two-dimensional, but the guards live on a single edge of it. We call this the \emph{{\sc Strait Guarding Problem}} ($\SGP$). The classical {\sc Art Gallery Problem} ($\AGP$) allows guards to stand anywhere inside the region and is $\exists\mathbb{R}$-complete, so no efficient exact algorithm is expected for it under any standard assumption. We show that pinning the guards to a single edge is not merely a realistic restriction. It is exactly the restriction that turns an intractable problem into one with a clean, exact, polynomial-time algorithm.

The \emph{{\sc Art Gallery Problem}} \cite{orourke1987artgallery}, introduced by Klee in 1973, asks a simple question: how many guards are needed to watch over an art gallery? Two points are \emph{mutually visible} if the straight-line segment joining them lies entirely inside $P$. Formally, for a simple polygon $P$ representing the floor plan of a gallery, a \emph{guard} is a point $g \in P$ that sees every point within its line of sight, and a \emph{guard set} is a collection of guards that together see every point of $P$. The $\SGP$ asks for a guard set of minimum size.

Chv\'{a}tal \cite{DBLP:journals/jctb/Chvatal75} proved the foundational combinatorial bound: $\lfloor n/3 \rfloor$ guards always suffice for an $n$-vertex polygon. The computational version, finding a minimum guard set for a given polygon, was shown to be \nph{} by Lee and Lin \cite{DBLP:journals/tit/LeeL86}. In a landmark result, Abrahamsen, Adamaszek, and Miltzow \cite{DBLP:journals/jacm/AbrahamsenAM22} established $\exists\mathbb{R}$-completeness, placing the problem strictly harder than any problem in \np{} under standard complexity assumptions. In response, a long line of work has developed approximation algorithms: Ghosh gave an $\OO(\log n)$-approximation for vertex and edge guards~\cite{GHOSH2010718}, and Bonnet and Miltzow gave an $\OO(\log \OPT)$-approximation for point guards under mild assumptions on the input \cite{bonnetmiltzow2017approx}. This sharp intractability makes identifying tractable subclasses a central goal in computational geometry. Recent work keeps finding such subclasses in unexpected places: requiring each guard to cover a single contiguous stretch of the boundary was very recently shown to make full polygon guarding solvable in polynomial time \cite{biniaz2025contiguous}, even though the unconstrained boundary-guarding problem is \nph{} \cite{laurentini1999guarding}, and in fact $\exists\mathbb{R}$-hard \cite{stade2025pointboundary}. Our paper adds a new tractable subclass to this list, this time by restricting where guards may stand rather than what they must cover.

\subsection{Weak Visibility Polygons.}

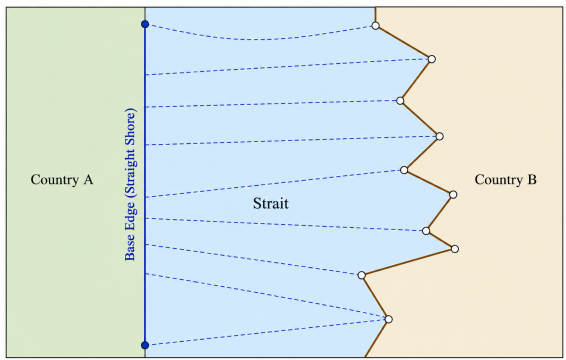
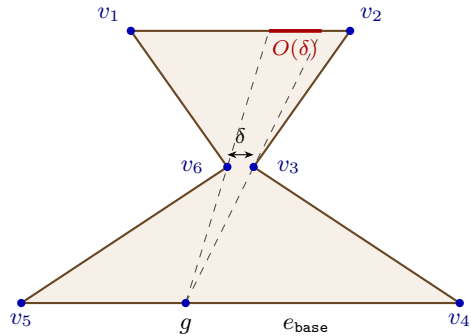
\begin{figure}[t]
\centering
\begin{subfigure}[b]{0.48\linewidth}
  \centering
  \includegraphics[width=\linewidth]{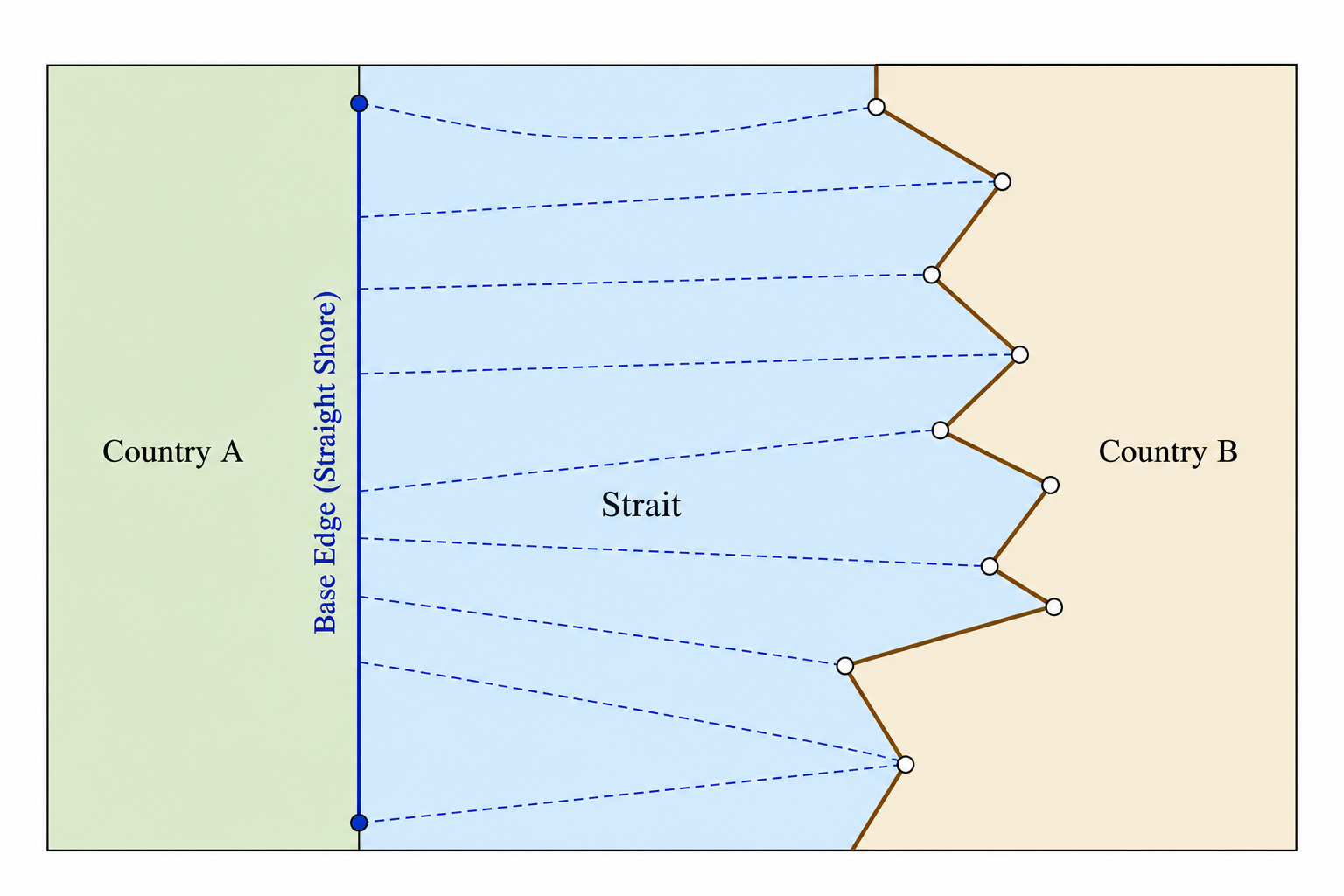}
  \subcaption{The {\sc Strait Guarding Problem}: guards placed on one straight
  shoreline, the base edge $\eb$, must monitor the entire strait.}\label{fig:strait-shore}
\end{subfigure}
\hfill
\begin{subfigure}[b]{0.48\linewidth}
  \centering
  \begin{tikzpicture}[xscale=2.9, yscale=1.8,
      vtx/.style={circle, fill=blue!70!black, inner sep=1.1pt},
      vlbl/.style={font=\footnotesize, text=blue!50!black},
      ray/.style={dashed, gray!60!black}]

    \coordinate (v1) at (-0.5, 2);
    \coordinate (v2) at ( 0.5, 2);
    \coordinate (v3) at ( 0.06, 1);
    \coordinate (v4) at ( 1, 0);
    \coordinate (v5) at (-1, 0);
    \coordinate (v6) at (-0.06, 1);

    \fill[brown!12] (v1) -- (v2) -- (v3) -- (v4) -- (v5) -- (v6) -- cycle;
    \draw[brown!55!black, thick] (v1) -- (v2) -- (v3) -- (v4) -- (v5) -- (v6) -- cycle;

    \draw[ray] (-0.25, 0) -- (0.13, 2);   
    \draw[ray] (-0.25, 0) -- (0.37, 2);   
    \node[vtx] at (-0.25, 0) {};
    \node[font=\footnotesize, below=2pt] at (-0.25, 0) {$g$};

    \draw[red!65!black, line width=1.4pt] (0.13, 2) -- (0.37, 2);
    \node[font=\scriptsize, text=red!65!black] at (0.25, 1.86) {$O(\delta)$};

    \draw[{Stealth[length=3.5pt]}-{Stealth[length=3.5pt]}, thin]
          (-0.06, 1.09) -- (0.06, 1.09)
          node[midway, above=0.5pt, font=\scriptsize] {$\delta$};

    \node[vtx] at (v1) {}; \node[vtx] at (v2) {}; \node[vtx] at (v3) {};
    \node[vtx] at (v4) {}; \node[vtx] at (v5) {}; \node[vtx] at (v6) {};
    \node[vlbl, above left]  at (v1) {$v_1$};
    \node[vlbl, above right] at (v2) {$v_2$};
    \node[vlbl] at ( 0.22, 1) {$v_3$};
    \node[vlbl, below] at (v4) {$v_4$};
    \node[vlbl, below] at (v5) {$v_5$};
    \node[vlbl] at (-0.22, 1) {$v_6$};

    \node[font=\footnotesize, below=2pt] at (0.3, 0) {$\eb$};
  \end{tikzpicture}
  \subcaption{A six-vertex $\wvp{}$ in which even the single edge $v_1 v_2$ requires arbitrarily many guards on $\eb$: every guard $g$ sees $v_1 v_2$ only through the gap of width $\delta$ between $v_6$ and $v_3$, so it sees a portion of length $O(\delta)$, and the number of guards grows as the gap narrows (\cref{prop:single-edge-unbounded}).}\label{fig:hourglass-intro}
\end{subfigure}
\caption{The two faces of the {\sc Strait Guarding Problem}: (a) the model, guards confined to one straight shoreline of a strait; (b) the difficulty, a six-vertex $\wvp{}$ whose optimum guard number is unbounded in $n$.}
\label{fig:strait-intro}
\end{figure}

One well-studied, practically motivated subclass is that of \emph{weak visibility polygons} ($\wvp${}s). The idea of visibility from a single edge goes back to Avis and Toussaint \cite{avis1981visibility}, who gave the first optimal algorithm for testing it. A simple polygon $P$ is a $\wvp{}$ with respect to an edge $e$, called the \emph{base edge} $\eb$, if every point of $P$ is visible from at least one point on $\eb$. Ghosh \cite{ghosh2007vis} gives a modern treatment of $\wvp{}$s and related visibility structures. A natural real-world instance is a \emph{coastal strait}: think of $\eb$ as one shoreline and the polygon as the navigable waterway enclosed between two coasts. Surveillance units deployed along that shore must collectively monitor the entire strait. Despite this structural restriction, $\SGP$ is non-trivial: the reflex vertices of $P$ cast shadow regions that interact intricately over $\eb$.

\Cref{fig:strait-intro} illustrates the two faces of the problem.
\Cref{fig:strait-shore} shows the strait itself: guards on one straight
shoreline, the base edge $\eb$, must monitor the entire waterway.
\Cref{fig:hourglass-intro} shows why this is substantially harder than
guarding an altitude terrain~\cite{DAESCU201922}, where one guard per terrain edge always
suffices. In the six-vertex hourglass polygon drawn there, every guard on
$\eb$ sees the top edge $v_1 v_2$ only through the gap of width $\delta$
between $v_6$ and $v_3$, so it sees a portion of that edge whose length is
proportional to $\delta$. Hence, the number of guards needed for this one
edge grows in proportion to $1/\delta$ as the gap narrows, while the polygon
keeps its six vertices, and no function of $n$ alone can bound the optimum.
We make this precise in \cref{prop:single-edge-unbounded}, Section ~\ref{sec:Algorithm}.

\subsection{Prior Work on WV-Polygon Guarding.}
Prior work on guarding $\wvp{}$s has mostly focused on approximation algorithms. Bhattacharya, Ghosh, and Roy \cite{BHATTACHARYA2017109} showed that guarding a $\wvp$ with point guards at arbitrary positions is \nph{}, and gave a $6$-approximation for vertex guarding; likewise, Ashur, Filtser, and Katz \cite{DBLP:journals/jocg/AshurFK21} improved this to a $(2+\varepsilon)$-approximation. A PTAS for vertex guarding was announced by Katz \cite{DBLP:journals/corr/abs-1803-02160} and developed in full generality by the authors of \cite{DBLP:journals/comgeo/AshurFKS22}, through a local-search technique for a class of graphs that includes the visibility graphs of both $\wvp{}$s and terrains. The authors of \cite{duraisamy2022halfguarding} studied a related \emph{half-guarding} model on $\wvp{}$s and terrains, in which each guard sees in one direction only, and the parameterized complexity of guarding simple polygons was studied in \cite{DBLP:journals/dcg/AgrawalKLSZ24}. To our knowledge, no \emph{exact polynomial-time} algorithm for guarding a $\wvp{}$ was known before this work, even for the restricted case of guards confined to $\eb$.

\subsection{Witness Sets.}

A fundamental lower-bounding tool is the \emph{witness set} \cite{DBLP:journals/ijcga/AmitMP10}: a set $W \subseteq P$ such that no single guard sees two distinct points of $W$ at once. Every witness set of size $k$ certifies that at least $k$ guards are necessary, so the largest such set, the \emph{witness number} $\mu_{WS}(P)$, lower-bounds the guard number. For general polygons, computing the witness number is hard. The authors of \cite{DBLP:journals/corr/abs-2605-01592} showed that the problem lies in $\mathsf{NP} \cap \mathsf{XP}$. They also studied its \emph{discrete} variant, in which witnesses must be chosen from a fixed finite set of points, and showed it to be \npc{}, even when the input is restricted to rectilinear polygons with holes; the authors of \cite{DBLP:journals/corr/abs-2511-10224} had already given a polynomial-time algorithm for it on simple polygons. For monotone polygons specifically, \cite{DBLP:journals/corr/abs-2511-10224} developed exact and approximation algorithms for the witness set problem itself, by exploiting the monotone structure.

\subsection{Monotone Mountain Polygons.}
A closely related class is that of \emph{monotone mountain polygons} \cite{DAESCU201922}: $x$-monotone polygons in which one boundary chain is a single horizontal segment, the base, and the other is an arbitrary monotone chain, the terrain, above it. Guarding a terrain with guards placed on the terrain itself is \nph{}~\cite{kingkrohn2011terrain}, so, here too, the complexity turns on where the guards may stand. The authors of \cite{DAESCU201922} studied \emph{altitude terrain guarding}, which places the minimum number of guards on a fixed horizontal line to see the entire terrain. They proved that every monotone mountain is \emph{perfect}, meaning its witness number equals its guard number, and gave an optimal $\OO(n)$ algorithm for the fixed-altitude case. They left open the following problem: given $k$, what is the minimum altitude $y^{*}$ at which $k$ guards on a horizontal line suffice to see the entire mountain? Very recently, Kang, Kim, and Ahn \cite{DBLP:conf/iwoca/KangKA25} solved this question, giving an $\OO(k^2 \lambda_{k-1}(n) \log n)$-time algorithm for even $k \geq 2$ and an $\OO(k^2 \lambda_{k-2}(n) \log n)$-time algorithm for odd $k \geq 3$, where $\lambda_s(n)$ is the length of the longest $(n, s)$-Davenport--Schinzel sequence \cite{DBLP:books/daglib/0080837}.

\subsection{Our Contributions.}
We study the {\sc Art Gallery Problem} for $\wvp{}$s with guards restricted to $\eb$. Our results stand in sharp contrast to the general $\exists\mathbb{R}$-hardness: for this natural case, we obtain an exact polynomial-time algorithm. The work is motivated by, and generalizes, that of \cite{DAESCU201922}.


\begin{enumerate}
  \item \textbf{Boundary--Interior Equivalence.}

  \vspace{-4pt}

  \begin{restatable}{theorem}{thmbdryequiv}\label{thm:boundary-equivalence}
    Guarding the boundary of a $\wvp$ is equivalent to guarding the entire polygon, when the guards are located at the base of the $\wvp$, i.e., $\gs(\bd(\wv),\eb) = \gs(\wv,\eb)$.
  \end{restatable} 

  \vspace{-8pt}
  
  \textbf{Brief Overview:} For any $\wvp{}$ $\wv$ with base $\eb$, we prove that any guard set covering $\bd(\wv)$ automatically covers the interior \textbf{(\cref{thm:boundary-equivalence}, \cref{sec:wvguardingequivalence})}. This equivalence is specific to guards confined to $\eb$; it does not extend to the general {\sc Art Gallery Problem} on $\wvp{}$s, where guards may lie anywhere inside the polygon \cite{DBLP:journals/jocg/AshurFK21}. Reducing the problem to a purely boundary task is the foundation of our approach and is key to our main result.
  
  \item \textbf{Finite Guardability and the Witness-Guard Algorithm.}

  Our \textit{main result} is the following.

  \vspace{-4pt}

  \begin{restatable}{theorem}{MainResult}\label{thm:main}
    The \textit{{\sc Strait Guarding Problem}}, $ \SGP $, for a weak visibility polygon $\wv$ can be solved in $\OO\bigl((n + \OPT\cdot\rho)\,(\log n + \log\OPT)\bigr)$ time, where $n = |V(\wv)|$, $\rho$ is the number of reflex vertices and $\OPT$ is the cardinality of the minimum-sized guard set.    
  \end{restatable}   

  \vspace{-8pt} 
  
  \textbf{Brief Overview:} The optimum guard number $\gs(\wv,\eb)$ is not bounded by any function of $n$ alone. We prove this using a six-edge $\wvp$ in \textbf{\cref{prop:single-edge-unbounded}}. Even so, we prove that every $\wvp$ is guardable from $\eb$ by finitely many guards; the argument rests on the compactness of a closed, bounded segment as a subspace of the line, through the Heine--Borel theorem \textbf{(\cref{lem:finite-guardable})}. Our algorithm rests on a witness set built from intervals on $\eb$: a collection of intervals of which no single point is pierced twice, so that no single guard sees two witnesses. This collection defines an interval graph, and by the perfectness of interval graphs \cite{GOLUMBIC20041}, the maximum independent set and the minimum clique cover have equal size. This equality allows us to derive a maximum witness set and output a minimum guard set together, thereby establishing optimality via a clean graph-theoretic argument rather than a separate proof. Concretely, we give the \emph{Witness-Guard Algorithm}: a greedy procedure that places each guard at the right endpoint of the uncovered interval with the leftmost right endpoint, then expands the candidate set by shooting rays through every reflex vertex the new guard sees. The algorithm outputs a minimum guard set of size exactly $\OPT$, together with a certifying witness set of equal size, in $\OO\bigl((n + \OPT\cdot\rho)\,(\log n + \log\OPT)\bigr)$ time \textbf{(\cref{lemma:opt}, \cref{lemma:complete}, \cref{thm:main}, Sections ~\ref{sec:Algorithm}, \ref{sec:algocorrectness}, \ref{sec:algoruntime})}.

  \item \textbf{A Tight $\Theta(n\log n)$ Bound for Vertex Guarding.}

  \vspace{-4pt}

  \begin{restatable}{theorem}{vsg}\label{thm:gs-lower-bound}
    For a $\wvp ~ \wv$, guarding the vertices of $\wv$ can be solved in $\Theta(n\log n)$ time.
  \end{restatable} 
    
  \vspace{-8pt}
  
  \textbf{Brief Overview:} The core of our technique is a geometric characterization of visibility on the base: using the two shortest-path trees rooted at the endpoints of $\eb$, each boundary point $p$ is assigned a closed interval $\I(p) \subseteq \eb$ such that a guard $g$ sees $p$ if and only if $g \in \I(p)$. This yields a $\Theta(n\log n)$ algorithm for guarding just the vertices of $\wv$: an $\OO(n \log n)$ upper bound from a minimum clique cover of the resulting interval graph, matched by an $\Omega(n \log n)$ lower bound that we prove by a reduction from Sorting \textbf{(\cref{thm:gs-lower-bound}, \cref{sec:vertexguardingWV})}.

  \item \textbf{Corollary: Optimal Perfect Guarding of an Altitude Terrain.}

  \vspace{-4pt}

  \begin{restatable}[Fixed-height perfect guarding]{theorem}{fhpg}\label{thm:fixed-height}
    Let $\mo$ be an $x$-monotone mountain, let $L$ be a fixed altitude line above it, and let $k$ be the minimum number of guards on $L$ that see all of $\mo$. Then a guard set $\{g_1^+,\dots,g_k^+\} \subseteq L$ of size $k$ guarding $\mo$, together with a certifying witness set $\{w_1,\dots,w_k\} \subseteq \partial\mo$ of the same size, can be computed in $\OO(n)$ time.
  \end{restatable}

  \vspace{-8pt}
  
  \textbf{Brief Overview:} As a first consequence for altitude terrains, we obtain an asymptotically optimal, linear-time algorithm that guards a terrain from a fixed horizontal line and certifies its own optimality: it returns a minimum guard set together with a matching witness set of equal size, in $\OO(n)$ time. This improves the previously best $\OO(n^{2}\log n)$ bound for producing a certifying witness set \textbf{(\cref{thm:fixed-height}, \cref{sec:fixed-height})} \cite{DAESCU201922}.

  \item \textbf{Corollary: Minimum-Altitude Terrain Guarding.}

  \vspace{-4pt}
  
  \begin{restatable}[Minimum-height guarding]{theorem}{mhg}\label{thm:main-result}
    Let $\mo$ be a $k$-guardable $x$-monotone mountain. For any integer $1 \le m \le k$, a height $y^\ast$, a guard set $G^\ast = \{g_1^+, \dots, g_m^+\} \subseteq L_{y^\ast}$ of size $m$ guarding all of $\mo$, and $m$ witnesses $W^\ast \subseteq \partial\mo$ certifying this guard number can be computed in $\OO(nk + k^2\log k)$ time, where $y^\ast$ is the least height at which $m$ guards suffice.
  \end{restatable}
  
  \vspace{-8pt}
  
  \textbf{Brief Overview:} We resolve the open problem of \cite{DAESCU201922}: given a $k$-guardable monotone mountain, one guardable by $k$ guards on a horizontal line, and a target $m \le k$, we compute the lowest altitude at which $m$ guards on a horizontal line suffice, together with their positions and matching witnesses certifying perfectness, in $\OO(nk + k^{2}\log k)$ time \textbf{(\cref{lem:single-step,thm:main-result}, \cref{sec:corollarykguard})}. This improves the $\OO(k^{2}\lambda_{k-1}(n)\log n)$ bound of Kang, Kim, and Ahn \cite{DBLP:conf/iwoca/KangKA25}, where $\lambda_{k-1}(n)$ is a near-linear Davenport--Schinzel length, through an elementary upward sweep that uses no such sequences.
\end{enumerate}


We close with two open directions that are revisited in the concluding section. First, can point guarding a $\wvp$ be solved in $\OO(n\log n)$ time? We guard every vertex of the polygon from the base in $\Theta(n\log n)$ time, but guarding the entire interior takes $\OO\bigl((n+\OPT\cdot\rho)(\log n+\log\OPT)\bigr)$ time; whether this can be reduced to $\OO(n\log n)$ remains open. Second, what changes if guards may be placed anywhere on $\bd(\wv)$ rather than only on $\eb$? 


\section{Preliminaries}\label{sec:prelims}
In this section, we set up the basic definitions and notation used throughout the paper. We start with simple polygons, monotone polygons, and terrains. We then define weak-visibility polygons and their base edges. Next, we define witness sets and guard sets, which are the two main tools we use later. Finally, we define anchors and the interval representation of a point. These last two ideas form the basis of our algorithm in Section~\ref{sec:Algorithm}.


\begin{definition}[\textbf{Simple Polygon}]\label{def:simplepolygon}
    A {\em simple polygon} $\po$ is the closed region of the plane bounded by a simple closed polygonal curve, that is, a curve formed by $n$ line segments in which consecutive segments share an endpoint and no two non-consecutive segments intersect~\cite{deberg2008computational}. These segments are the {\em edges} of $\po$, and their shared endpoints are its {\em vertices}. We write $V(\po)$ and $E(\po)$ for the sets of vertices and edges of $\po$, respectively.
\end{definition}
    A simple polygon $\po$ encloses a region, called its {\em interior}, that has a measurable area. We use $ \bd(\po), ~\inte(\po), ~\ex(\po) $ to denote the boundary, interior, and exterior region of $\po$, respectively.  A vertex $v \in V(\po)$ is a {\em reflex}  vertex if the interior angle of $\po$ at $v$ is larger than 180 degrees, else it is called a {\em convex} vertex.

\begin{definition}[\textbf{Visibility Polygon}]\label{def:visibility}
    In a simple polygon $\po$, two points $x$ and $y$ see each other (i.e., mutually visible) if the line joining them lies entirely inside of $\po$. For a point $x \in \po$, the visibility polygon of $x$, is defined as:
    \[\vis(x) \coloneqq \{y \in \po~ | ~y~\text{is visible from}~x \}\]
\end{definition}
    
\begin{definition}[\textbf{Monotone Polygon}]\label{def:monotonepolygon}
    A polygon $\mo$ in the plane is called {\em monotone} with respect to a straight line $L$, if every line orthogonal to $L$ intersects the boundary of $\mo$ at most twice.
\end{definition}
    Throughout our article, we refer to a polygon that is monotone with respect to the line $y=0$ (or the $x$-axis) as an \textit{$x$-monotone polygon}. For any point $p$ in the plane,  $x(p)$ denotes the $x$-coordinate of $p$.

\begin{definition}[\textbf{Terrain}]\cite{DAESCU201922}\label{def:Terrain}
    A terrain is an $x-$monotone polygonal chain.
\end{definition}    

\begin{definition}[\textbf{Altitude Line}]\cite{DAESCU201922}\label{def:AltitudeLine}
    Given a terrain $\mathcal{T}$, an altitude line is a horizontal line segment lying entirely above it, that is, the $y-$coordinate of the line is greater than that of all the points in the terrain $\mathcal{T}$.
\end{definition}


\begin{definition}[\textbf{Monotone Mountain Polygon}]\cite{DBLP:conf/cccg/ORourke97}\label{def-monotone_mountain}
   A monotone polygon is called a \textit{Monotone Mountain} if one of the chains is a single edge. 
\end{definition}

\begin{definition}[\textbf{Weak Visibility Polygon or $\wvp$}] \label{def-wvpolygon}
    A simple polygon $P$ is a \emph{weak visibility polygon} (WV polygon) with respect to an edge $e = (u, v)$ (the \emph{base}) if every point $p \in P$ is visible from at least one point on $e$. \cite{ghosh2007vis}. 
\end{definition}

Throughout our paper, we will use the notation $\eb$ to denote the base edge.

\begin{definition}[\textbf{Witness Set}]\label{def:WitnessSet}
    A set of points $W \subseteq \po$ is said to be a {\em witness set} \cite{DBLP:journals/ijcga/AmitMP10} in  $\po$ if for every pair of points $w_1,w_2 \in W$, there exists no guard $g$ such that both $w_1$ and $w_2$ are visible from~$g$. In other words, every point $p \in \po$ is visible from at most one witness $w \in W$. We denote a {\em witness set} by $\ws(\po, Q)$ where the polygon considered is $\po$ and the witnesses are from the set $Q$. The cardinality of a {\em witness set} is denoted by $\witN(\po, Q)$ (see \cref{fig:WVPWitness}).
\end{definition}

\begin{figure}[ht!]
    \centering
    \includegraphics[scale=0.9]{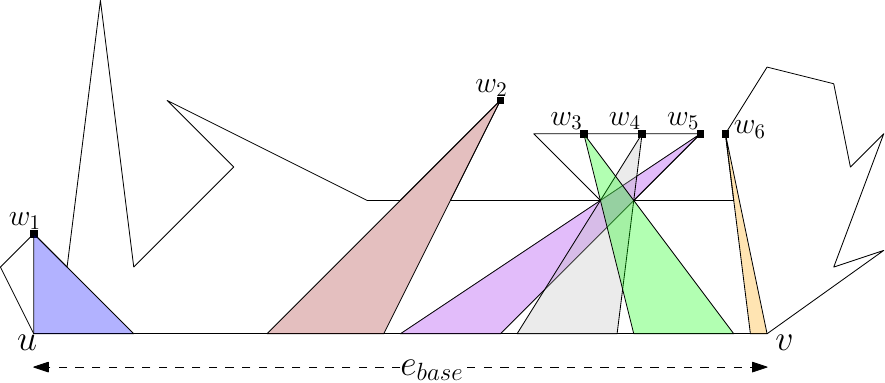}
    \caption{Witness Set Illustration. A witness set $W = \{w_1,\dots,w_6\}$ in $\po$ with pairwise disjoint visibility regions, giving $\witN(\po,Q) = 6$.}
    \label{fig:WVPWitness}
\end{figure}

\begin{definition}[\textbf{Guard Set}]\label{def:GuardSet}
    A set of points $G \subseteq \po$ is a {\em guard set}~\cite{DBLP:journals/jacm/AbrahamsenAM22} if every point $p \in \po$ is visible from at least one point $g \in G$. If the guards may lie anywhere in the polygon, we call them {\em point guards}; if they are restricted to the vertices, we call them {\em vertex guards}. We write $\gs(\po, A)$ for a guard set of $\po$ whose guards are drawn from a set $A \subseteq \po$, and $\grN(\po, A)$ for its cardinality. Thus a point-guard set is $\gs(\po, \po)$ and a vertex-guard set is $\gs(\po, V(\po))$ (see \cref{fig:WVPGuards}).
\end{definition}

\begin{figure}[H]
    \centering
    \includegraphics[scale=0.9]{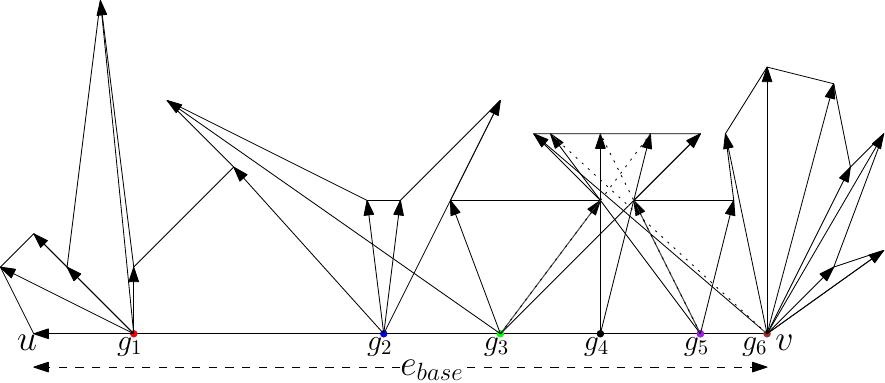}
    \caption{Guard Set Illustration. A guard set $G = \{g_1,\dots,g_6\}$ in $\po$ whose visibility regions jointly cover $\po$, giving $\grN(\po, A) = 6$.}
    \label{fig:WVPGuards}
\end{figure}

If $\po$ is a $\wvp$, with $\eb = (u,v)$ as its base edge, then for any point $p \in \po$, we can compute its {\em shortest path} from $p$ to the vertices $u$ and $v$. We denote the shortest paths from $u$ and $v$ by $\spl(p)$ and $\spr(p)$, respectively.  

We then arrive at a crucial definition motivated by \cite{DBLP:journals/corr/abs-2511-10224}, which will be useful throughout our paper. 

\begin{definition}[Anchors]
For any point $p \in \po$:
\begin{itemize}
  \item The \emph{right anchor} $\RA(p)$ is the first reflex vertex in the shortest path $\spr(p)$.
  \item The \emph{left anchor} $\LA(p)$ is the first reflex vertex in the shortest path $\spl(p)$.
\end{itemize}
In other words, $\RA(p)$ (resp.\ $\LA(p)$) is the last reflex vertex on the shortest path from $v$ (resp.\ $u$) to $p$ inside $\po$, or the base endpoint itself if no reflex vertex intervenes, i.e., the shortest path in that case is a straight line (see \cref{fig:RA_LA_Int}).
\end{definition}

\begin{definition}[Interval Representation]\label{def-interval}
For a point $p$ in a $\wvp$ $\po$, join the line from $p$ to $\RA(p)$ and extend it until it hits the base $\eb$, yielding a point $r(p) \in \eb$ (\cref{fig:RA_LA_Int}). Similarly, draw the line from $p$ to $\LA(p)$ and extend it until it hits $\eb$, yielding $\ell(p) \in \eb$. If $\LA(p)$ (resp. $\RA(p)$) does not exist then we take $\ell(p) = u$ (resp. $r(p) = v$). The \emph{interval of $p$} is: $\I(p) = [\ell(p),\, r(p)] \subseteq \eb$.

\end{definition}

\begin{figure}[H]
    \centering
    \includegraphics[scale=1.05]{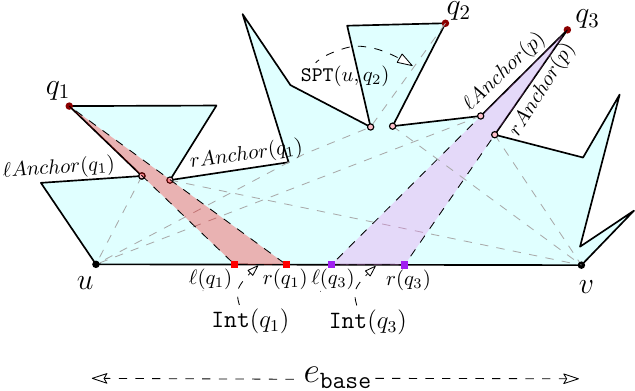}
    \caption{This figure depicts an example of $\RA, ~\LA, ~\I()$ construction for points in a weak visibility polygon.} \label{fig:RA_LA_Int}
\end{figure}


\section{Weak Visibility Polygons and Interval Graph} \label{sec:intervalconstruction}

In this section, we show that in a $\wvp$, the visibility regions, when restricted to the base edge, form intervals, and their intersection graph is an \textit{interval graph} \cite{GOLUMBIC20041}. We begin by providing a few preliminary definitions and observations.

\vspace{2mm}

Consider a pair of points $x,y$ in a polygon $\po$, which cannot see each other, i.e., $y \notin \vis(x)$, and vice-versa. Then, the line joining $x$ and $y$, must intersect the exterior $\ex(\po)$ of $\po$. We look at any point where the line $\overline{xy}$ \textit{exits} $\po$, while traversing from $x$ to $y$. We name such a point on $\bd(\po)$ as $h_{\mathtt{out}}$, and we look at the immediate next point where it \textit{enters} the polygon again, and name the point on $\bd(\po)$ as $h_{\mathtt{in}}$. Now, we arrive at the following definition.


    \begin{definition}[hill]\label{hill} \cite{DBLP:journals/corr/abs-2511-10224}
        {\em Consider a pair of points $x,y$ in a polygon $\po$, which cannot see each other, i.e., $y \notin \vis(x)$, and vice-versa. Let $h_{\mathtt{out}}$ be a point where $\overline{xy}$ exits $\po$, and let $h_{\mathtt{in}}$ be the immediately following point where $\overline{xy}$ re-enters $\po$. Both $h_{\mathtt{out}}$ and $h_{\mathtt{in}}$ lie on $\overline{xy}$, and the open segment between them lies in $\ex(\po)$. The points $h_{\mathtt{out}}$ and $h_{\mathtt{in}}$ split $\bd(\po)$ into two chains; together with the segment $\overline{h_{\mathtt{out}}h_{\mathtt{in}}}$, each chain forms a simple closed curve enclosing a bounded region. For exactly one of the two chains, the interior of this enclosed region is contained in $\ex(\po)$; we call that chain the \textit{hill} corresponding to the mutually invisible pair $x$ and $y$ (the red chain in \cref{reflex}).}
    \end{definition}

    The immediate observation concerns when two points in $\po$ cannot see each other. Refer to \cref{reflex} for an illustration.
    
\begin{observation}\label{obs-hill}
     Consider a pair of {distinct} points $x,y$ in a polygon $\po$. If they are not visible to each other, then there must exist a hill, corresponding to the pair of points $x,y$ (note that there might exist multiple hills corresponding to $x,y$). Any such hill (which is a chain in $\bd(\po)$) contains at least one reflex vertex.
\end{observation}

\begin{proof}
    From the definition of a hill, it cannot contain all convex vertices, for then it would mean that the line $\overline{h_{\mathtt{out}}h_{\mathtt{in}}} \subseteq \overline{xy}$ lies inside $\po$, which violates the definition of the points $h_\mathtt{out}$ and $h_\mathtt{in}$.
\end{proof}

\begin{figure}[ht!]
    \centering
    \includegraphics[scale=0.9]{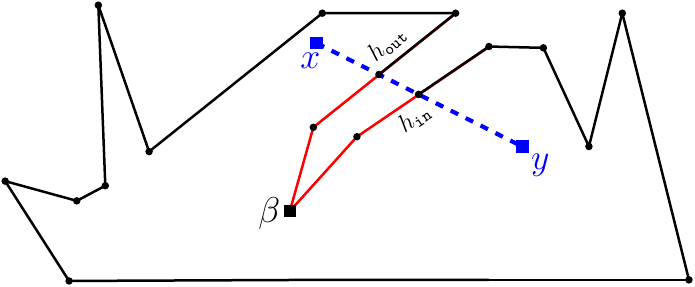}
    \caption{The points $x$ and $y$ cannot see each other. The chain (part of $\bd(\po)$) from $h_\mathtt{out}$ to $h_\mathtt{in}$ (marked in red) is called a \textit{hill} corresponding to $x,y$. $\beta$ is a reflex vertex on the hill.} \label{reflex}
\end{figure}

\begin{proposition}\label{prop:intervalvisibility}
    If $\wv$ is a $\wvp$, and $p$ is any point in $\wv$, then its interval (as in \cref{def-interval}), $\I(p) = \vis(p) \cap \eb.$
\end{proposition}




\begin{proof}
    Since $\wv$ is a weak visibility polygon, the point $p$ is visible from some point of $\eb$, so $\vis(p) \cap \eb$ is nonempty. Being the intersection of the two closed sets $\vis(p)$ and $\eb$, it is closed.

    It remains to show that $\vis(p) \cap \eb$ is an interval, i.e., convex as a subset of the line supporting $\eb$. We argue geometrically. Let $a, c \in \vis(p) \cap \eb$ with $a \neq c$, and let $b$ be any point of $\eb$ strictly between $a$ and $c$. We claim that $p$ sees $b$. The segments $\overline{pa}$ and $\overline{pc}$ lie in $\wv$, because $p$ sees $a$ and $c$, and $\overline{ac} \subseteq \eb \subseteq \wv$, since $\eb$ is an edge of $\wv$. Hence, the boundary of the triangle $\triangle pac$ is contained in $\wv$.

    Suppose, for contradiction, that some point of $\triangle pac$ lies in the exterior $\ex(\wv)$. Such a point lies in the open interior of the triangle, since the triangle's boundary is contained in $\wv$. By the Jordan curve theorem, $\ex(\wv)$ is connected, so it contains a path from this point to a point lying outside $\triangle pac$. This path leaves the triangle and therefore meets its boundary; but the boundary lies in $\wv$ and is thus disjoint from $\ex(\wv)$, a contradiction. Therefore $\triangle pac \subseteq \wv$, and in particular $\overline{pb} \subseteq \triangle pac \subseteq \wv$, so $p$ sees $b$.

    Thus $\vis(p) \cap \eb$ is a nonempty closed interval. Its endpoints are the extreme points of $\eb$ visible from $p$, namely $\ell(p)$ and $r(p)$ by \cref{def-interval}: the point $p$ sees no point of $\eb$ to the left of $\ell(p)$ or to the right of $r(p)$, and it sees both $\ell(p)$ and $r(p)$. Consequently $\vis(p) \cap \eb = [\ell(p), r(p)] = \I(p)$.
\end{proof}

\begin{proposition}[Visibility Equivalence]\label{prop:vis-equiv}
A guard $g \in \eb$ sees $p \in \wv$ if and only if $g \in \I(p)$.
\end{proposition}

\begin{proof}
If $g \in \I(p)$, then by \cref{prop:intervalvisibility}, which tells us that $\I(p) = \vis(p) \cap \eb$, $g \in \vis(p)$. So, $g$ sees $p$. Conversely, if $g \in \eb$ sees $p$, then $g \in \vis(p)$. This implies that $g \in \vis(p) \cap \eb = \I(p)$, as required.
\end{proof}


\section{Guarding a Weak Visibility Polygon}\label{sec:wvguardingequivalence}

In this section, we show that guarding a $\wvp ~\wv$ is equivalent to guarding \textit{only} its boundary, $\bd(\wv)$. This turns a two-dimensional covering problem into a one-dimensional one, and this fact will be used in the subsequent sections. The idea is motivated by the paper of Ashur, Filtser, and Katz \cite{DBLP:journals/jocg/AshurFK21}.

Recall that by a $\wvp$, $\wv$ (as in \cref{def-wvpolygon}), we mean that every point in $\wv$ is visible from some edge of the polygon, called $\eb$. We try to provide some positive answers to the question of guarding $\wv$ where the guard locations are restricted to the edge $\eb$. In other words, can we provide a solution to the problem of computing $\gs(\wv,\eb)$? 



\begin{definition}[Interval Graph $\mathcal{G}_Q$]
Given a discrete point set $Q \subseteq \wv$, form the interval graph $\mathcal{G}_Q = (Q, E)$ where
\[
  \{w_1, w_2\} \in E \iff \I(w_1) \cap \I(w_2) \neq \emptyset.
\]
\end{definition}

Since interval graphs are perfect:
\[
  \alpha(\mathcal{G}_Q) = \bar{\chi}(\mathcal{G}_Q),
\]
where $\alpha$ is the maximum independent set size (= maximum witness set size) and $\bar{\chi}$ is the minimum clique cover number.

We make the following claims, which are presented in the subsequent subsections.  

\subsection{Equivalence of Guarding Models}

\thmbdryequiv*

\begin{proof}
Let $\wv$ be a weak visibility polygon with respect to a base edge $\eb$, and let $G$ be a set of guards placed on $\eb$ that collectively guard the entire boundary $\bd(\wv)$. We prove that $G$ also guards the entire interior of $\wv$. 

A maximal connected subset $H$ of $\wv$ is called a {\em hole} in $\wv$ if no point in $H$ is visible from any guard in $G$. It is known that any region of $\wv$ not directly visible from the guard set (i.e., a hole) forms a convex pocket in the interior of $\wv$ (Observation 4, Section 2.3 in \cite{DBLP:journals/jocg/AshurFK21}).

Consider any point $x \in \inte(\wv)$. If $x$ lies on the boundary, it is already guarded by assumption. Otherwise, $x$ lies either in a region directly visible from the base or inside such a convex pocket.

Since $\wv$ is weakly visible from $\eb$, there exists a point $g \in \eb$ such that the line segment $\overline{gx}$ lies entirely within $\wv$. Extend the segment $\overline{gx}$ beyond $x$ until it first intersects the boundary $\bd(\wv)$ at a point $b$ (see \cref{fig:interior-hole}).

By assumption, the boundary is fully guarded, so there exists a guard $g_b \in G$ such that $b$ is visible from $g_b$, i.e., the segment $\overline{g_b b}$ lies entirely within $\wv$. So:
\begin{itemize}
    \item $\overline{b g} \subseteq \wv$ (since $g$ sees $x$ and $b$ lies on the extension),
    \item $\overline{b g_b} \subseteq \wv$ (since $g_b$ sees $b$).
\end{itemize}

\begin{figure}
    \centering
    \includegraphics[width=.5\linewidth]{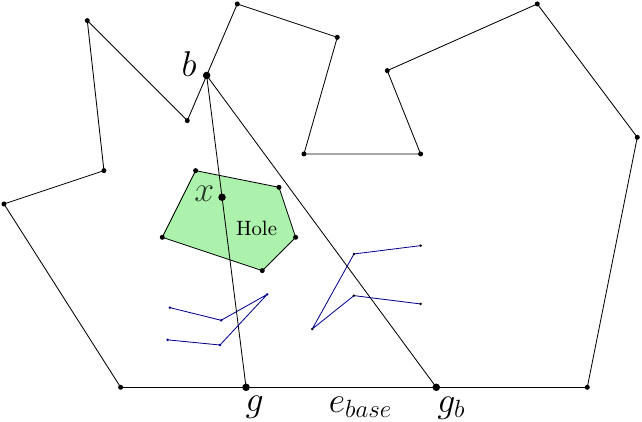}
    \caption{Illustration of point $x$ inside a convex pocket, with $g$ on the base and $g_b$ guarding boundary point $b$.}
    \label{fig:interior-hole}
\end{figure}

We now claim that $x$ is visible from $g_b$. To establish this, consider the triangle $\triangle b g g_b$.

We already know:
\begin{itemize}
    \item $\overline{b g} \subseteq \wv$,
    \item $\overline{b g_b} \subseteq \wv$.
\end{itemize}

Suppose, for contradiction, that $x$ is not visible from $g_b$. Then the segment $\overline{g_b x}$ must intersect the boundary of $\wv$, implying the existence of an obstruction (a hill) within $\wv$ that blocks visibility.

However, since $\wv$ is a simple polygon, any such hill must intersect at least one of the segments $\overline{b g}$ or $\overline{b g_b}$, contradicting the fact that both these segments lie entirely within $\wv$.

Thus, no such obstruction exists, and $g_b$ sees every point on the segment $\overline{b g}$, including $x$.

Therefore, every interior point $x \in \wv$ is visible from some guard in $G$, implying that $G$ guards the entire interior of $\wv$.

Hence, boundary guarding from the base implies guarding the entire polygon from the base. The converse direction is trivial, completing the proof.
\end{proof}  


\section{The Algorithm}\label{sec:Algorithm}

In this section, we present an algorithm that addresses a special case of the {\sc Art Gallery Problem}.

Our Problem Statement ($\SGP$) is defined as follows.

\begin{tcolorbox}[enhanced, title={\color{black} {\sc $k-$Guard Set in $\wvp$}: \textbf{Decision Variant}},
    colback=white, boxrule=0.4pt,
    attach boxed title to top left={xshift=6pt, yshift*=-3.5mm},
    boxed title style={size=small, frame hidden, colback=white},
    before skip=8pt, after skip=8pt]
    \textbf{Input:} A {\em weak visibility} polygon $\wv$, with its base edge $\eb$, and an integer $k$.\\
    \textbf{Task:}\hspace*{1mm} Is it possible to guard $\wv$ from the base using at most $k$ guards? If yes, give a minimum-sized guard set $\gs(\wv,\eb)$.
\end{tcolorbox}

We also address the following optimization variant, in which no bound $k$ is prescribed.

\begin{tcolorbox}[enhanced, title={\color{black} {\sc Minimum Guard Set in $\wvp$: \textbf{Optimization Variant}}},
    colback=white, boxrule=0.4pt,
    attach boxed title to top left={xshift=6pt, yshift*=-3.5mm},
    boxed title style={size=small, frame hidden, colback=white},
    before skip=8pt, after skip=8pt]
    \textbf{Input:} A {\em weak visibility} polygon $\wv$ with its base edge $\eb$.\\
    \textbf{Task:}\hspace*{1mm} Find a minimum-sized guard set $\gs(\wv,\eb)$ on $\eb$ that guards all of $\wv$.
\end{tcolorbox}

The optimization variant is well-posed: as we establish in Lemma~\ref{lem:finite-guardable}, every $\wvp$ is finitely guardable from its base. Although finitely guardable, the number of guards needed can be arbitrarily large, and we cannot give any specific bound on the maximum number of guards needed based only on the input size $n$. Our algorithm, however, is output-sensitive and yields the optimal guard set in poly$(n, k)$ time if $k$ is the size of $G_{\OPT}$. Consequently, Algorithm~\ref{alg:base-guard}, when run without the $k$-step stopping condition, is guaranteed to terminate in finitely many iterations and outputs a minimum-sized guard set.

We give a solution to the above-stated problems. Before getting into the details of the algorithm, we define a particular notion regarding visibility.

\begin{definition}[Strong Visibility Interval]\label{def:strongvisibility}
    If $\wv$ is a {\em weak visibility} polygon with the base edge $\eb$, then, for any edge $e$ of $\wv$, the strong visibility interval of $e$, denoted by $\SI(e)$, is a subset of $\eb$ such that for any point $x \in \SI(e)$, $x$ sees the entire edge $e$.
\end{definition}

\begin{observation}
    If $e = (u_1,u_2)$ is an edge of a $\wvp, ~\wv$, then its strong visibility interval is the intersection of the intervals of the two vertices of the edge, i.e., $\SI(e) = \I(u_1) \cap \I(u_2)$. 
\end{observation}

\begin{proof}
    If $\I(u_1) \cap \I(u_2) = \emptyset$, then there exists no point on $\eb$ which simultaneously sees both the vertices of the edge $e$. Thus, there exists no point on $\eb$ that sees the entire edge $e$ (since such a point must see both the vertices $u_1$ and $u_2$), which implies that $\SI(e) = \emptyset$. Otherwise, if $\I(u_1) \cap \I(u_2) \neq \emptyset$, then every point on $\I(u_1) \cap \I(u_2)$ sees both $u_1$ and $u_2$, and hence, sees the entire edge $e = (u,v)$. Thus, $\I(u_1) \cap \I(u_2) \subseteq \SI(e)$. Conversely, any point $x \in \SI(e)$ sees the entire edge $e$, so it sees both the vertices $u_1$ and $u_2$. Thus, by \cref{prop:vis-equiv}, $x \in \I(u_1)$ and $x \in \I(u_2)$. Therefore, $x \in \I(u_1) \cap \I(u_2)$. So, we get that $\I(u_1) \cap \I(u_2) = \SI(e)$, as required.
\end{proof}

Observe that {\em Monotone Mountain}s (\cref{def-monotone_mountain}) are a proper subclass of $\wvp$s with the property that every edge has a non-empty strong visibility interval.

The authors in \cite{DAESCU201922} gave an $\OO(n)$ algorithm for guarding altitude terrains ($\ATG$), where every edge is guaranteed to have a strong visibility interval on the base, a property inherent to altitude terrains. $\wvp$s, however, do not enjoy this luxury. An edge may have no strong visibility interval on the base at all, making guarding even one edge of a $\wvp{}$ to require arbitrarily many guards on the base edge, in sharp contrast to altitude terrains, where a single guard per edge always suffices.

Before stating our algorithm, we first prove the following proposition (\ref{prop:single-edge-unbounded}) using a six-vertex $\wvp$, which shows why the $\WVP$ case is significantly harder than $\ATG$ and supports our unbounded claim about the number of guards needed.


\begin{proposition}\label{prop:single-edge-unbounded}
For every real $\delta$ with $0 < \delta < 1/2$ there is a simple polygon $P_\delta$ with six vertices, weakly visible from an edge $\eb$ of its boundary, that contains an edge $AB$ of length $1$ with the following property: every set $G$ of points on $\eb$ such that each point of $AB$ is seen from some point of $G$ satisfies $|G| \ge \lceil 1/(2\delta) \rceil$. In particular, the number of guards on the base edge needed to see even a single edge of an $n$-vertex $\wvp{}$ cannot be bounded by any function of $n$.
\end{proposition}

\begin{proof}
Define $P_\delta$ to be the hexagon with vertices
$A = (-\tfrac12, 2)$, $B = (\tfrac12, 2)$, $C = (\tfrac{\delta}{2}, 1)$, $D = (1, 0)$, $E = (-1, 0)$, and $F = (-\tfrac{\delta}{2}, 1)$, in this cyclic order, and let $\eb = DE$ (see \cref{fig:hourglass}). The vertices $C$ and $F$ are its only reflex vertices, and $P_\delta$ is the union of the two convex trapezoids $Q^- = EDCF$ and $Q^+ = FCBA$, which meet exactly in the segment $FC$; we call $FC$ the \emph{window}, and its length is $\delta$. The polygon is symmetric under the reflection $(x, y) \mapsto (-x, y)$. Moreover,
\begin{equation}\label{eq:cross-section}
P_\delta \cap \{\, y = 1 \,\} \;=\; FC ,
\end{equation}
because the line $y = 1$ meets $Q^-$ only in its top side $FC$ and meets $Q^+$ only in its bottom side $FC$.

\emph{The polygon $P_\delta$ is weakly visible from $\eb$.}
Every point of $Q^-$ sees every point of $\eb$, because $Q^-$ is convex and contains $\eb$. Every point $(u, 1)$ of the window is seen from the point $(u, 0) \in \eb$ directly below it. Now let $q = (u, v)$ be a point of $Q^+$ with $v > 1$; by symmetry, we may assume that $u \le 0$. As $\omega$ ranges over $[-\delta/2, \delta/2]$, the ray from $q$ through the window point $(\omega, 1)$ meets the line $y = 0$ at the point $(\varphi(\omega), 0)$, where
\[
\varphi(\omega) \;=\; \frac{\omega v - u}{v - 1}
\]
is an increasing function of $\omega$, so the feet of these rays form the interval $[\varphi(-\delta/2), \varphi(\delta/2)]$. Since $u \le 0$, we get $\varphi(\delta/2) > 0 > -1$. Since $q$ lies on or to the right of the line through $F$ and $A$, we have $-u \le \delta/2 + (v-1)(1-\delta)/2$, and substituting this bound gives $\varphi(-\delta/2) \le (1 - 2\delta)/2 < 1$. The interval $[\varphi(-\delta/2), \varphi(\delta/2)]$ therefore intersects $[-1, 1]$, so some point $p \in \eb$ lies on a ray from $q$ through a window point $w$. By convexity the segment $pw$ lies in $Q^-$ and the segment $wq$ lies in $Q^+$, so $pq \subseteq P_\delta$; hence $p$ sees $q$.

\emph{A single guard sees a portion of $AB$ of length at most $2\delta$.}
Fix a point $p = (x, 0)$ on $\eb$ and let $q$ be a point of $AB$. The segment $pq$ rises from height $0$ to height $2$, so it crosses the line $y = 1$ in exactly one point; if $pq \subseteq P_\delta$, then this crossing point lies in the window by \cref{eq:cross-section}. Conversely, if $pq$ crosses the window at a point $w$, then $pw \subseteq Q^-$ and $wq \subseteq Q^+$ by convexity, so $pq \subseteq P_\delta$. Hence $p$ sees exactly those points of $AB$ that lie on rays from $p$ through window points. The ray from $p$ through $(\omega, 1)$ meets the line $y = 2$ at the point $(2\omega - x,\, 2)$, so the portion of $AB$ seen from $p$ is the segment
\begin{equation}\label{eq:visible-part}
V(p) \;=\; \bigl( [\, -\delta - x,\; \delta - x \,] \times \{2\} \bigr)
\,\cap\, AB .
\end{equation}

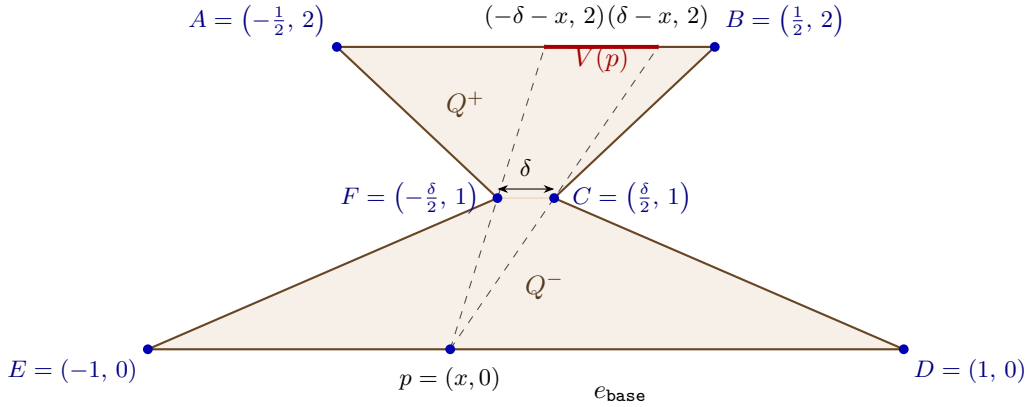
\begin{figure}[H]
\centering
\begin{tikzpicture}[xscale=5, yscale=2,
    vtx/.style={circle, fill=blue!70!black, inner sep=1.3pt},
    vlbl/.style={font=\footnotesize, text=blue!50!black},
    ray/.style={dashed, gray!60!black}]

  \pgfmathsetmacro{\del}{0.15}
  \pgfmathsetmacro{\hw}{\del/2}        
  \pgfmathsetmacro{\gx}{-0.2}          
  \pgfmathsetmacro{\vl}{-\del-\gx}     
  \pgfmathsetmacro{\vr}{\del-\gx}      

  \coordinate (A) at (-0.5, 2);
  \coordinate (B) at ( 0.5, 2);
  \coordinate (C) at ( \hw, 1);
  \coordinate (D) at ( 1, 0);
  \coordinate (E) at (-1, 0);
  \coordinate (F) at (-\hw, 1);

  \fill[brown!12] (A) -- (B) -- (C) -- (D) -- (E) -- (F) -- cycle;
  \draw[brown!55!black, thick] (A) -- (B) -- (C) -- (D) -- (E) -- (F) -- cycle;
  \draw[brown!40] (F) -- (C);   

  \node[font=\small, text=brown!45!black] at (-0.16, 1.62) {$Q^{+}$};
  \node[font=\small, text=brown!45!black] at ( 0.05, 0.42) {$Q^{-}$};

  \draw[ray] (\gx, 0) -- (\vl, 2);   
  \draw[ray] (\gx, 0) -- (\vr, 2);   
  \node[vtx] at (\gx, 0) {};
  \node[font=\footnotesize, below=3pt] at (\gx, 0) {$p = (x, 0)$};

  \draw[red!65!black, line width=1.5pt] (\vl, 2) -- (\vr, 2);
  \draw[red!65!black] (\vl, 1.985) -- (\vl, 2.015)
                      (\vr, 1.985) -- (\vr, 2.015);
  \node[font=\footnotesize, above=3pt] at (\vl, 2) {$(-\delta - x,\, 2)$};
  \node[font=\footnotesize, above=3pt] at (\vr, 2) {$(\delta - x,\, 2)$};
  \node[font=\small, text=red!65!black] at ({(\vl+\vr)/2}, 1.92) {$V(p)$};

  \draw[{Stealth[length=4pt]}-{Stealth[length=4pt]}, thin]
        (-\hw, 1.06) -- (\hw, 1.06)
        node[midway, above=1pt, font=\footnotesize] {$\delta$};

  \node[vtx] at (A) {}; \node[vtx] at (B) {}; \node[vtx] at (C) {};
  \node[vtx] at (D) {}; \node[vtx] at (E) {}; \node[vtx] at (F) {};
  \node[vlbl, above left]  at (A) {$A = \bigl(-\tfrac12,\, 2\bigr)$};
  \node[vlbl, above right] at (B) {$B = \bigl(\tfrac12,\, 2\bigr)$};
  \node[vlbl, right=3pt]   at (C) {$C = \bigl(\tfrac{\delta}{2},\, 1\bigr)$};
  \node[vlbl, below right] at (D) {$D = (1,\, 0)$};
  \node[vlbl, below left]  at (E) {$E = (-1,\, 0)$};
  \node[vlbl, left=3pt]    at (F) {$F = \bigl(-\tfrac{\delta}{2},\, 1\bigr)$};

  \node[font=\small, below=10pt] at (0.25, 0) {$\eb$};
\end{tikzpicture}
\caption{The hourglass polygon $P_\delta$ of
  \cref{prop:single-edge-unbounded}, drawn with $\delta = 0.15$. A guard $p = (x, 0)$ on $\eb$ sees the edge $AB$ only through the window $FC$; the portion $V(p)$ that it sees is the image of the window under the central projection from $p$ and has length at most $2\delta$.}
\label{fig:hourglass}
\end{figure}

The window image $[-\delta - x, \delta - x]$ has length exactly $2\delta$, independently of the position of $p$; only its placement along the line $y = 2$ depends on $x$. Consequently, the length of $V(p)$ is at most $2\delta$.

\emph{Counting.}
Let $G$ be a set of $k$ points on $\eb$ that collectively see every point of $AB$. Then $AB = \bigcup_{p \in G} V(p)$, and comparing lengths in \cref{eq:visible-part} yields
\[
1 \;=\; |AB| \;\le\; \sum_{p \in G} |V(p)| \;\le\; 2 k \delta ,
\]
so $k \ge 1/(2\delta)$, and hence $k \ge \lceil 1/(2\delta) \rceil$ because $k$ is an integer. Finally, $P_\delta$ has six vertices for every $\delta$, while $\lceil 1/(2\delta) \rceil \to \infty$ as $\delta \to 0$; therefore, no function of the number of vertices alone bounds the number of guards needed for the edge $AB$.
\end{proof}

Our goal is to solve the guarding problem for WV-polygons \emph{without} assuming the property mentioned in \cref{prop:single-edge-unbounded}. We propose an algorithm in the next subsection to address this problem.

\begin{algorithm}[H]
\caption{Optimal Base-Edge Guarding of a Weak Visibility Polygon}
\label{alg:base-guard}
\DontPrintSemicolon
\SetKwInOut{Input}{Input}
\SetKwInOut{Output}{Output}
\SetKw{KwAnd}{and}
\SetKw{KwTo}{to}
\SetKwFunction{ComputeInterval}{ComputeInterval}
\SetKwFunction{ShootRay}{ShootRay}
\Input{A weak visibility polygon $\wv$ with base edge $\eb$.}
\Output{A minimum guard set $G \subseteq \eb$.}
\BlankLine
\tcp{--- Initialization ---}
$Q_0 \leftarrow V(\wv)$\tcp*{\textcolor{blue}{all vertices of $\wv$}}
$G_0 \leftarrow \emptyset$\;
\ForEach{$q \in Q_0$}{
    compute $\I(q) \subseteq \eb$\;
}
\BlankLine
\tcp{--- Main loop ---}
$i \leftarrow 1$\;
\Repeat{$\mathcal{U}_i = \emptyset$}{
    \BlankLine
    \tcp{Collect uncovered candidate intervals}
    $\mathcal{U}_i \leftarrow \bigl\{\, \I(q) : q \in Q_{i-1},\;\; \I(q) \cap G_{i-1} = \emptyset ~or ~\I(q) \cap G_{i-1} = \{g_{i-1}\} \, \text{with} ~\ell(q) = g_{i-1} < r(q) \bigr\}$\;
    \If{$\mathcal{U}_i \neq \emptyset$}{
        \BlankLine
        \tcp{Select the candidate with the leftmost right endpoint}
        $q^{*} \leftarrow \displaystyle\arg\min_{q\,:\;\I(q)\,\in\,\mathcal{U}_i}\; r\!\left(\I(q)\right)$\;
        $g_i \leftarrow r\!\left(\I(q^{*})\right)$\;
        $G_i \leftarrow G_{i-1} \cup \{g_i\}$\;
        \BlankLine
        \tcp{Witness expansion: join reflex vertices visible from $g_i$ and extend it so that it hits $\bd(\wv)$}
        $Q_i \leftarrow Q_{i-1}$\;
        \ForEach{reflex vertex $v_j$ of $\wv$ visible from $g_i$}{
            $b_j \leftarrow$ intersection of line $\overrightarrow{g_i v_j}$ with $\bd(\wv)$
                             beyond $v_j$\;
            $Q_i \leftarrow Q_i \cup \{b_j\}$\;
            compute $\I(b_j) \subseteq \eb$\;
        }
        $i \leftarrow i + 1$\;
    }
}
\KwRet{$G_{i-1}$}\tcp*{\textcolor{blue}{$G_{i-1}$ is the optimal guard set}}
\end{algorithm}

\subsection{Description of the Algorithm}
    We give a brief description of our algorithm (see \cref{fig:algo-run} for a step-by-step illustration of \cref{alg:base-guard}).
\subsubsection{Invariant notation} 
We maintain:
\begin{itemize}
  \item $Q_i$: a discrete set of \emph{witness candidates} in $\wv$ (points whose intervals are drawn on $\eb$).
  \item $G_i$: a set of guards placed on $\eb$.
\end{itemize}


\tikzset{
  sgpvtx/.style={circle, fill=blue!70!black, inner sep=1.1pt},
  sgpvlbl/.style={font=\tiny, text=blue!50!black, inner sep=1.5pt},
  sgpwit/.style={circle, fill=red!70!black, inner sep=1.2pt},
  sgpwlbl/.style={font=\tiny, text=red!65!black, inner sep=1.5pt},
  sgpgrd/.style={rectangle, fill=green!40!black, inner sep=1.7pt},
  sgpglbl/.style={font=\tiny, text=green!35!black, inner sep=1.5pt},
  sgpqsel/.style={draw=black, circle, inner sep=2.5pt},
  sgpray/.style={dashed, gray!55!black, line width=0.5pt,
                 -{Stealth[length=3.5pt]}},
  sgpguide/.style={densely dotted, green!40!black, line width=0.4pt},
}

\newcommand{\SGPpoly}{%
  \coordinate (E)  at (-1,0);    \coordinate (D)  at (4,0);
  \coordinate (Cp) at (3.1,1);   \coordinate (Bp) at (3,2);
  \coordinate (Ap) at (1.6,2);   \coordinate (Fp) at (1.4,1);
  \coordinate (C)  at (0.15,1);  \coordinate (B)  at (0.5,2);
  \coordinate (A)  at (-0.5,2);  \coordinate (F)  at (-0.15,1);
  \coordinate (g1) at (-0.2,0);  \coordinate (g2) at (0.4,0);
  \coordinate (g3) at (3.8,0);
  \coordinate (w1) at (-0.1,2);  \coordinate (w2) at (2.4,2);
  \coordinate (w3) at (-0.1818,1.0909);
  \fill[brown!10]
    (E)--(D)--(Cp)--(Bp)--(Ap)--(Fp)--(C)--(B)--(A)--(F)--cycle;
  \draw[brown!55!black, thick]
    (E)--(D)--(Cp)--(Bp)--(Ap)--(Fp)--(C)--(B)--(A)--(F)--cycle;
  \draw[brown!35] (F)--(C) (Fp)--(Cp);
  \foreach \p in {E, D, Cp, Bp, Ap, Fp, C, B, A, F}{\node[sgpvtx] at (\p) {};}
  \node[sgpvlbl, above]       at (A)  {$A$};
  \node[sgpvlbl, above]       at (B)  {$B$};
  \node[sgpvlbl, above]       at (Ap) {$A'$};
  \node[sgpvlbl, above]       at (Bp) {$B'$};
  \node[sgpvlbl, below left]  at (F)  {$F$};
  \node[sgpvlbl, below right=1pt] at (C)  {$C$};
  \node[sgpvlbl, below right] at (Fp) {$F'$};
  \node[sgpvlbl, below right] at (Cp) {$C'$};
  \node[sgpvlbl, below]       at (E)  {$E$};
  \node[sgpvlbl, below]       at (D)  {$D$};
  \node[font=\tiny, below] at (2.05,-0.02) {$\eb$};
}

\newcommand{\ivU}[5]{%
  \draw[#4, line width=1.1pt] (#2,#1) -- (#3,#1);
  \draw[#4, line width=0.7pt]
    (#2,#1-0.05)--(#2,#1+0.05) (#3,#1-0.05)--(#3,#1+0.05);
  \node[font=\tiny, text=#4, left=1pt] at (#2,#1) {#5};}
\newcommand{\ivC}[5]{%
  \draw[#4!28, line width=1.1pt] (#2,#1) -- (#3,#1);
  \draw[#4!28, line width=0.7pt]
    (#2,#1-0.05)--(#2,#1+0.05) (#3,#1-0.05)--(#3,#1+0.05);
  \node[font=\tiny, text=#4!45, left=1pt] at (#2,#1) {#5};}

\newcommand{\grdmark}[3]{\node[sgpgrd] at (#1,0) {};
  \node[sgpglbl, #3] at (#1,0) {#2};}
\newcommand{\selmark}[2]{%
  \node[rectangle, rotate=45, fill=black, inner sep=1.2pt]
    at (#1,#2) {};}

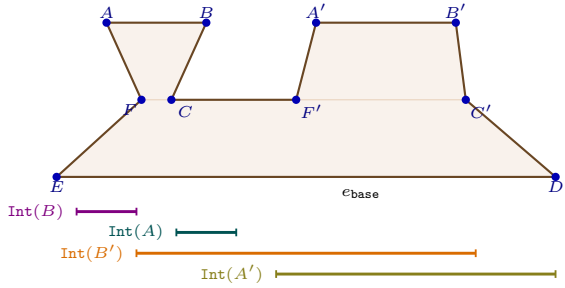
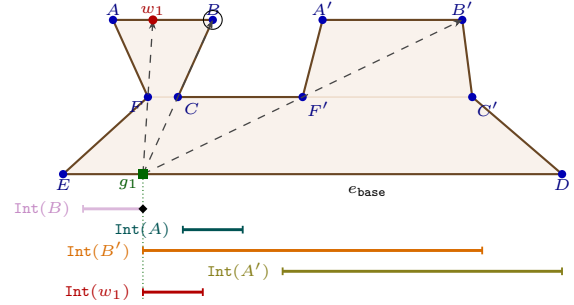
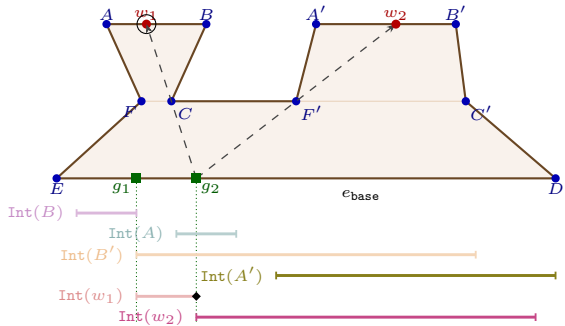
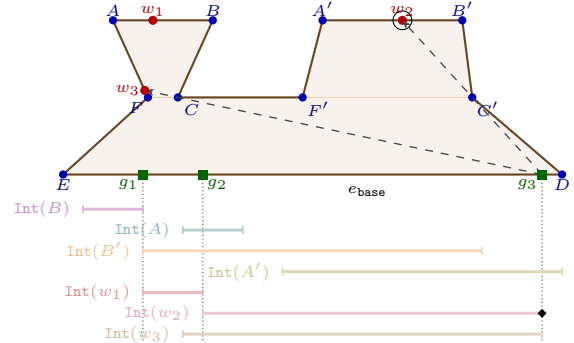
\begin{figure}[t]
\centering
\begin{subfigure}[t]{0.49\textwidth}
\centering
\begin{tikzpicture}[xscale=1.32, yscale=1.02]
  \SGPpoly
  \ivU{-0.45}{-0.8}{-0.2}{violet}{$\I(B)$}
  \ivU{-0.72}{0.2}{0.8}{teal!65!black}{$\I(A)$}
  \ivU{-0.99}{-0.2}{3.2}{orange!85!black}{$\I(B')$}
  \ivU{-1.26}{1.2}{4}{olive!85!black}{$\I(A')$}
\end{tikzpicture}
\subcaption{Initialization: $Q_0 = V(\wv)$ and $G_0 = \emptyset$. Only the four nontrivial visibility intervals are shown; every remaining vertex lies on the boundary of the convex lower chamber $EDC'F$, so its interval is all of $\eb$.}
\label{fig:algo-run-a}
\end{subfigure}
\hfill
\begin{subfigure}[t]{0.49\textwidth}
\centering
\begin{tikzpicture}[xscale=1.32, yscale=1.02]
  \SGPpoly
  \node[sgpqsel] at (B) {};
  \draw[sgpray] (g1) -- (w1);   
  \draw[sgpray] (g1) -- (B);    
  \draw[sgpray] (g1) -- (Bp);   
  \node[sgpwit] at (w1) {};
  \node[sgpwlbl, above] at (w1) {$w_1$};
  \grdmark{-0.2}{$g_1$}{below left}
  \draw[sgpguide] (-0.2,-0.06) -- (-0.2,-1.62);
  \ivC{-0.45}{-0.8}{-0.2}{violet}{$\I(B)$}
  \selmark{-0.2}{-0.45}
  \ivU{-0.72}{0.2}{0.8}{teal!65!black}{$\I(A)$}
  \ivU{-0.99}{-0.2}{3.2}{orange!85!black}{$\I(B')$}
  \ivU{-1.26}{1.2}{4}{olive!85!black}{$\I(A')$}
  \ivU{-1.53}{-0.2}{0.4}{red!70!black}{$\I(w_1)$}
\end{tikzpicture}
\subcaption{Iteration $1$: the interval $\I(B)$ has the leftmost right endpoint, so $q^{*} = B$ and $g_1 = r(\I(B))$. Extending $\protect\overrightarrow{g_1 F}$, $\protect\overrightarrow{g_1 C}$, and $\protect\overrightarrow{g_1 F'}$ beyond the reflex vertices creates the witness $w_1 \in AB$ and rehits $B$ and $B'$. Both $\I(w_1)$ and $\I(B')$ meet $G_1$ only in their left endpoint $g_1$, so they remain uncovered.}
\label{fig:algo-run-b}
\end{subfigure}

\medskip

\begin{subfigure}[t]{0.49\textwidth}
\centering
\begin{tikzpicture}[xscale=1.32, yscale=1.02]
  \SGPpoly
  \node[sgpwit] at (w1) {};
  \node[sgpwlbl, above] at (w1) {$w_1$};
  \node[sgpqsel] at (w1) {};
  \draw[sgpray] (g2) -- (w1);   
  \draw[sgpray] (g2) -- (w2);   
  \node[sgpwit] at (w2) {};
  \node[sgpwlbl, above] at (w2) {$w_2$};
  \grdmark{-0.2}{$g_1$}{below left}
  \grdmark{0.4}{$g_2$}{below right}
  \draw[sgpguide] (-0.2,-0.06) -- (-0.2,-1.89);
  \draw[sgpguide] (0.4,-0.06) -- (0.4,-1.89);
  \ivC{-0.45}{-0.8}{-0.2}{violet}{$\I(B)$}
  \ivC{-0.72}{0.2}{0.8}{teal!65!black}{$\I(A)$}
  \ivC{-0.99}{-0.2}{3.2}{orange!85!black}{$\I(B')$}
  \ivU{-1.26}{1.2}{4}{olive!85!black}{$\I(A')$}
  \ivC{-1.53}{-0.2}{0.4}{red!70!black}{$\I(w_1)$}
  \selmark{0.4}{-1.53}
  \ivU{-1.80}{0.4}{3.8}{magenta!80!black}{$\I(w_2)$}
\end{tikzpicture}
\subcaption{Iteration $2$: $q^{*} = w_1$ and $g_2 = r(\I(w_1))$; the intervals $\I(A)$ and $\I(B')$ are now stabbed in their interiors. Extending $\protect\overrightarrow{g_2 F'}$ creates the witness $w_2 \in A'B'$, whose interval is grazed by $g_2$ at its left endpoint.}
\label{fig:algo-run-c}
\end{subfigure}
\hfill
\begin{subfigure}[t]{0.49\textwidth}
\centering
\begin{tikzpicture}[xscale=1.32, yscale=1.02]
  \SGPpoly
  \node[sgpwit] at (w1) {};
  \node[sgpwlbl, above] at (w1) {$w_1$};
  \node[sgpwit] at (w2) {};
  \node[sgpwlbl, above] at (w2) {$w_2$};
  \node[sgpqsel] at (w2) {};
  \draw[sgpray] (g3) -- (w2);   
  \draw[sgpray] (g3) -- (w3);   
  \node[sgpwit] at (w3) {};
  \node[sgpwlbl, left] at (w3) {$w_3$};
  \grdmark{-0.2}{$g_1$}{below left}
  \grdmark{0.4}{$g_2$}{below right}
  \grdmark{3.8}{$g_3$}{below left}
  \draw[sgpguide] (-0.2,-0.06) -- (-0.2,-2.16);
  \draw[sgpguide] (0.4,-0.06) -- (0.4,-2.16);
  \draw[sgpguide] (3.8,-0.06) -- (3.8,-2.16);
  \ivC{-0.45}{-0.8}{-0.2}{violet}{$\I(B)$}
  \ivC{-0.72}{0.2}{0.8}{teal!65!black}{$\I(A)$}
  \ivC{-0.99}{-0.2}{3.2}{orange!85!black}{$\I(B')$}
  \ivC{-1.26}{1.2}{4}{olive!85!black}{$\I(A')$}
  \ivC{-1.53}{-0.2}{0.4}{red!70!black}{$\I(w_1)$}
  \ivC{-1.80}{0.4}{3.8}{magenta!80!black}{$\I(w_2)$}
  \selmark{3.8}{-1.80}
  \ivC{-2.07}{0.2}{3.8}{brown!75!black}{$\I(w_3)$}
\end{tikzpicture}
\subcaption{Iteration $3$: $q^{*} = w_2$ and $g_3 = r(\I(w_2))$. Extending $\protect\overrightarrow{g_3 C}$ creates $w_3 \in FA$, but $\I(w_3)$ already contains $g_2$ in its interior; hence $\mathcal{U}_4 = \emptyset$ and the algorithm outputs $G = \{g_1, g_2, g_3\}$.}
\label{fig:algo-run-d}
\end{subfigure}
\caption{A complete run of Algorithm~\ref{alg:base-guard} on a ten-vertex $\wvp$ with $\lvert G_{\OPT}\rvert = 3$. The left chamber is the hourglass $P_\delta$ of \cref{prop:single-edge-unbounded} with $\delta = 3/10$, so no single guard sees all of $AB$; the right chamber has a wide window and is guarded by a single point. Visibility intervals are drawn below $\eb$: solid while uncovered, faded once covered by the current guard set. Dotted vertical lines mark guard positions, a black diamond marks $r(\I(q^{*}))$, and the circled point is the selected witness $q^{*}$. Expansion rays whose extensions immediately leave $\wv$ are omitted.}
\label{fig:algo-run}
\end{figure}
\FloatBarrier

\subsubsection{Initialization}
\[
  Q_0 := V(\wv) \quad \text{(all vertices of }\wv\text{)}, \qquad G_0 := \emptyset.
\]

\subsubsection{Step \texorpdfstring{$i \geq 1$}{i >= 1}: Guard Placement and Witness Expansion}

\paragraph{Guard selection.}
Define the set of \emph{uncovered} candidate intervals:
\[
  \mathcal{U}_i := \{ \I(q) : q \in Q_{i-1},\; \I(q) \cap G_{i-1} = \emptyset ~or ~\{g_{i-1}\} \}. 
\]
If $\mathcal{U}_i = \emptyset$, terminate; $G_{i-1}$ is the output guard set.

Otherwise, select:
\[
  q^* := \arg\min_{q \in \mathcal{U}_i} r(q)
  \qquad \text{(the candidate with the leftmost right endpoint)}
\]
and place a guard:
\[
  g_i := r(q^*), \qquad G_i := G_{i-1} \cup \{g_i\}.
\]

\paragraph{Witness expansion.}
Join $g_i$ with every reflex vertex $v_j$ of $\wv$ that is visible from $g_i$, and extend it towards the boundary of $\wv$. Let each such line hit $\bd(\wv)$ at a boundary point $b_j$. Define:
\[
  \Delta Q_i := \{ b_j : v_j \text{ reflex},\; v_j \text{ visible from } g_i \},
\]
\[
  Q_i := Q_{i-1} \cup \Delta Q_i.
\]
Compute $\I(b_j)$ for each new point and add them to the interval collection.

\subsubsection{Termination}

\paragraph{Decision variant.}
When a bound $k$ is prescribed, the algorithm runs for at most $k$ iterations, terminating with the guard set
\[
  G_{k} = \{g_1, g_2, \ldots, g_{k}\}.
\]

\paragraph{Optimization variant.}
When the algorithm is run without the $k$-step stopping condition, i.e.\ it continues until $\mathcal{U}_i = \emptyset$, termination in finitely many steps is guaranteed by the following lemma, whose proof uses the compactness of closed intervals on the base.

We now show that the Optimization Variant of the Algorithm \ref{alg:base-guard} indeed stops after finitely many steps.

\begin{lemma}[Finite Guardability of a $\wvp$]\label{lem:finite-guardable}
  Every $\wvp\ \wv$ with base edge $\eb$ can be guarded from $\eb$ by finitely many guards.
\end{lemma}
\begin{proof}
It suffices to show that each edge $e$ of $\wv$ is individually guardable by finitely many guards on $\eb$: if every edge $e_i$ ($1\le i\le n$) admits a guard set of size $k_{e_i}<\infty$, then $n\cdot\max_i k_{e_i}$ guards on $\eb$ suffice globally. Fix $e = [u,v]$ and write $\eb = [a,b]$.

For each $x \in \eb$, let $\I(x) \cap e = [x_1,x_2]$ (possibly empty) denote the visibility interval of $x$ restricted to $e$, and define
\[
  V_x \;=\;
  \begin{cases}
    (x_1,\, x_2)   & \text{if } x_1 \neq u \text{ and } x_2 \neq v,\\[3pt]
    [x_1,\, x_2)   & \text{if } x_1 = u    \text{ and } x_2 \neq v,\\[3pt]
    (x_1,\, x_2]   & \text{if } x_1 \neq u \text{ and } x_2 = v,\\[3pt]
    [x_1,\, x_2]   & \text{if } x_1 = u    \text{ and } x_2 = v,
  \end{cases}
\]
setting $V_x = \emptyset$ when $\I(x)\cap e = \emptyset$. Each $V_x$ is open in the subspace topology \cite{munkres_topology_2018} on $e$: it equals $U\cap e$ for an open $U\subseteq\mathbb{R}$ obtained by extending $V_x$ slightly past $u$ or $v$ wherever an endpoint coincides with $u$ or $v$.

\medskip
\noindent\textbf{Covering.}\;
We show that for every $y \in e$, there exists $x \in \eb$ with $y \in V_x$.

For each $x \in \eb$ with $\I(x) \cap e \neq \emptyset$, write $\I(x) \cap e = [\ell(x),\, r(x)]$, where $\ell(x)$ and $r(x)$ are the left and right endpoints of the portion of $e$ visible from $x$. Unpacking the definition of $V_x$, the condition $y \in V_x$ holds if and only if
\[
  \ell(x) < y < r(x),
  \quad\text{or}\quad y = u = \ell(x),
  \quad\text{or}\quad y = v = r(x).
\]

\smallskip
\noindent\textit{Endpoints.}\;
For any $x_0 \in \eb$ seeing $u$: since $u$ is the left endpoint of $e$, we have $\ell(x_0) = u$, so $u \in [u, r(x_0)) \subseteq V_{x_0}$. Symmetrically, $v \in V_x$ for any $x$ seeing $v$.

\smallskip
\noindent\textit{Interior points.}\;
Fix $y \in \mathrm{int}(e)$ and let
\[
  S_y \;=\; \{x \in \eb : x \text{ sees } y\}.
\]
Since $\wv$ is a $\wvp$, $S_y \neq \emptyset$; $S_y = [a_y, b_y]$ is a closed subinterval of $\eb$.

We want some $x^* \in S_y$ satisfying $\ell(x^*) < y < r(x^*)$. The guards in $S_y$ that could fail this condition are those for which $y$ coincides exactly with the left or right endpoint of their visibility interval on $e$:
\[
  B_y \;=\; \{x \in S_y : y = \ell(x) \;\text{or}\; y = r(x)\}.
\]

\noindent\textit{$B_y$ is finite.}\;
Fix any $x \in B_y$ with $y = \ell(x) \neq u$, so that $y$ is the left endpoint of $\I(x) \cap e$ and $y \neq u$. This means the visibility of $e$ to the left of $y$ from $x$ is blocked by some vertex $p$ of $\wv$: that is, $x$, $p$, and $y$ are collinear. For a fixed vertex $p$ and fixed $y$, the line through $p$ and $y$ meets $\eb$ in at most one point; hence each vertex of $\wv$ contributes at most one guard to $B_y$. The same reasoning applies to the condition $y = r(x)$. Therefore $|B_y| \leq n$.

\noindent\textit{$S_y$ has positive length.}\;
Suppose for contradiction that $S_y = \{x^*\}$.
Then the visibility of $y$ both opens and closes at the same guard position $x^*$ on $\eb$. The left boundary of $S_y$ at $x^*$ (if $x^*$ is not the left endpoint $a$ of $\eb$) is caused by a vertex $q_L$ of $\wv$ with $x^*$, $q_L$, $y$ collinear. The right boundary (if $x^*$ is not the right endpoint $b$ of $\eb$) is caused by a vertex $q_R$ of $\wv$ with $x^*$, $q_R$, $y$ collinear.

If $q_L \neq q_R$, then $q_L,\, q_R,\, x^*,\, y$ are four collinear points, which is excluded by the general position assumption.

If $q_L = q_R = q$, then a single vertex $q$ simultaneously opens and closes visibility of $y$. As $x$ sweeps through $x^*$, the line through $x$ and $q$ can only transition the visibility of $y$ in one direction (either $y$ enters or leaves $\I(x)\cap e$), so a single vertex cannot cause both a left and a right boundary at the same $x^*$.

In either case, we reach a contradiction, so $b_y > a_y$ and $S_y$ has positive length.

\smallskip
\noindent\textit{Conclusion.}\;
Since $S_y = [a_y, b_y]$ is an uncountable set and $B_y$ is finite, we may choose $x^* \in S_y \setminus B_y$. Then $x^*$ sees $y$ (so $y \in [\ell(x^*), r(x^*)]$), and $y \neq \ell(x^*)$ and $y \neq r(x^*)$ (since $x^* \notin B_y$), giving $y \in (\ell(x^*), r(x^*)) \subseteq V_{x^*}$.

Hence $\{V_x\}_{x \in \eb}$ covers $e$.

\medskip
\noindent\textbf{Compactness.}\;
$e$ is a closed bounded interval, hence compact; by the Heine--Borel theorem~\cite{munkres_topology_2018}, the open cover $\{V_x\}_{x\in\eb}$ admits a finite subcover $V_{x_1},\ldots,V_{x_{m_e}}$ for some $x_1,\dots,x_{m_e}\in\eb$.

\medskip
\noindent\textbf{Guard set.}\;
For any $y\in e$, we have $y\in V_{x_j}$ for some $j$, hence $y\in\I(x_j)\cap e$ and $x_j$ sees $y$. Thus $\{x_1,\dots,x_{m_e}\}$ guards $e$, giving $k_e\le m_e<\infty$.
\end{proof}

\begin{corollary}\label{cor:opt-termination}
  Algorithm~\ref{alg:base-guard}, when run without the $k$-step stopping condition, terminates after at most $n \cdot k^{*}$ iterations and outputs a minimum-sized guard set $\gs(\wv,\eb)$, where $k^{*} = \max_{e} m_e$ is the maximum finite subcover size taken over all edges of $\wv$.
\end{corollary}


\section{Correctness}\label{sec:algocorrectness}

\subsection{Part I: Optimality (Lower Bound Matching Upper Bound)}

\begin{lemma}\label{lemma:opt}
At termination after $k-$steps, $|G_{k}|$ equals the minimum guard number for $\wv$ with respect to $\eb$. 
\end{lemma}


\begin{proof}
\textbf{Lower bound.} We claim that the algorithm implicitly constructs a witness set $W$ of size $k$, where $q_i^*$ denotes the candidate point selected at step $i$, and $g_i = r(q_i^*)$ denotes the guard placed at step $i$.

\begin{definition}[Clone of a point]%
\cite{DBLP:journals/corr/abs-2511-10224, DBLP:journals/corr/abs-2605-01592}%
\label{def:clonept}
Let $\po$ be a polygon and let $p$ be a point on the boundary $\bd(\po)$. A \emph{clone point} of $p$, denoted $\clvs[]{p}$, is a point on $\bd(\po)$ that is \emph{infinitesimally} close to $p$.
\end{definition}

Recall the selection rule at step $i$: $q_i^*$ attains the leftmost right endpoint among the intervals in
\[
    \mathcal{U}_i
    \;=\;
    \bigl\{\,\I(q) : q \in Q_{i-1},\;\I(q) \cap G_{i-1} = \emptyset
    ~\text{or}~ \{g_{i-1}\}\,\bigr\}.
\]
So $q_i^*$ falls into one of two cases.
\begin{itemize}
    \item \textbf{Case 1.}\label{case1point} $\I(q_i^*) \cap G_{i-1} = \emptyset$.
    \item \textbf{Case 2.}\label{case2point} $\I(q_i^*) \cap G_{i-1} = \{g_{i-1}\}$, so $\ell(q_i^*) = g_{i-1} = r(q_{i-1}^*)$.
\end{itemize}
By definition of $\mathcal{U}_i$, Case 2 can only ever touch the \emph{most recent} guard $g_{i-1}$, never an earlier one.

We also use a monotonicity fact about the selection, already noted earlier: since the algorithm always takes the admissible interval with the leftmost right endpoint,
\[
    r(q_j^*) < r(q_i^*), \qquad \text{for all } j < i. \tag{$\ast$}
\]

\medskip
\noindent\textbf{Step 1: How the intervals $\I(q_i^*)$ sit relative to each other.}

\smallskip
\noindent\textbf{Claim.} \emph{Let $j<i$. Then either}
\begin{enumerate}
    \item[(a)] $\I(q_j^*) \cap \I(q_i^*) = \emptyset$, \emph{or}
    \item[(b)] $i = j+1$, $q_i^*$ \emph{is a Case 2 point, and} $\I(q_j^*) \cap \I(q_i^*) = \{g_j\}$.
\end{enumerate}

\smallskip
\noindent\textit{Proof of Claim.} Suppose $q_i^*$ is a Case 1 point. Then $g_j \notin \I(q_i^*)$ for every $j<i$. If we had $r(q_j^*) \ge \ell(q_i^*)$, then together with $(\ast)$ we would get $r(q_j^*) \in \I(q_i^*)$, a contradiction. So $r(q_j^*) < \ell(q_i^*)$, and $\I(q_j^*) \cap \I(q_i^*) = \emptyset$. This gives (a) for every $j<i$.

Suppose instead $q_i^*$ is a Case 2 point. Then $g_j \notin \I(q_i^*)$ for every $j < i-1$, and the same argument gives $r(q_j^*) < \ell(q_i^*)$, so (a) holds for these $j$. For $j = i-1$, we already have $\ell(q_i^*) = g_{i-1} = r(q_{i-1}^*)$, so $\I(q_{i-1}^*)$ ends exactly where $\I(q_i^*)$ begins: the two intervals share exactly the point $g_{i-1}$. This is (b). \hfill $\diamond$

So the intervals of $q_1^*,\ldots,q_k^*$ never overlap. The only place two of them can even touch is between consecutive points, and only when the latter one is a Case 2 point; there they meet at exactly one point, $g_{i-1}$.

\medskip
\noindent\textbf{Step 2: The freedom in choosing a clone.}

A clone $\clvs{p}$ is, by \cref{def:clonept}, any point close enough to $p$ on $\bd(\wv)$; the definition does not pin down exactly how close. We use this freedom throughout. Away from the finitely many vertices where the reflex vertex defining $\ell(\cdot)$ or $r(\cdot)$ changes, both maps vary continuously along $\bd(\wv)$. So if $r(q_j^*) < \ell(q_i^*)$ with some positive gap, then for $\clv[j]{q}$ close enough to $q_j^*$, we still have $r(\clv{q}) < \ell(q_i^*)$: as $\clv{q}$ walks from $q_j^*$ towards its candidate clone position, $r(\clv[j]{q})$ moves continuously away from $r(q_j^*)$, so we can always stop early enough, at some point $t^*$ on this stretch of $\bd(\wv)$, to keep $r(t^*)$ inside the gap $(r(q_j^*), \ell(q_i^*))$. The same holds symmetrically on the $\ell(\cdot)$ side. Below, each clone will need to satisfy only finitely many such requirements (at most $k-1$, one for each other witness), and each requirement is satisfied simply by choosing a clone sufficiently close to its source point. So we can always choose it close enough to satisfy all requirements together; we fix exactly how close only once all requirements on it are known.

\medskip
\noindent\textbf{Step 3: Constructing the witness $w_i$.}

If $q_i^*$ is a Case 1 point, set $w_i = q_i^*$.

If $q_i^*$ is a Case 2 point, set $w_i$ to be a clone $\clv[i]{q}$, chosen as follows. Since $g_{i-1} = \ell(q_i^*)$, the guard $g_{i-1}$ sees only a contiguous arc of $\bd(\wv)$ with $q_i^*$ at one end. So one side of $q_i^*$ is not visible from $g_{i-1}$, and we place $\clv[i]{q}$ there.

\begin{itemize}
    \item \textit{Sub-case 2a: $q_i^*$ is a vertex of $\wv$.}\label{ProofWitnessCase1} The endpoint $\ell(q_i^*)$ comes from joining $q_i^*$ to the reflex vertex $\LA(q_i^*)$ and extending to $\eb$. Placing $\clv[i]{q}$ on the side of $q_i^*$ away from this ray moves the point out of $g_{i-1}$'s visible arc (\cref{figcaseA1}).
    \item \textit{Sub-case 2b: $q_i^*$ lies in the interior of an edge $e=(u_1,u_2)$.}\label{ProofWitnessCase2} The guard $g_{i-1}$ sees only a contiguous portion of $e$, with $q_i^*$ at one endpoint of it. Placing $\clv[i]{q}$ on the other side of $q_i^*$ along $e$ again removes it from $g_{i-1}$'s visible portion (\cref{figcaseA2}).
\end{itemize}
In either sub-case, $\I(\clv[i]{q})$ is a contiguous interval on $\eb$ that no longer contains its former left endpoint $g_{i-1}$, so
\[
    \ell(\clv[i]{q}) > g_{i-1}. \tag{$\ast\ast$}
\]
Write $\varepsilon_i := \ell(\clv[i]{q}) - g_{i-1} > 0$ for this gap.

After $k$ steps we obtain $W = \{w_1,\ldots,w_k\}$, with $w_i \in \{q_i^*,\ \clv[i]{q}\}$.

\begin{figure}[ht!]
     \centering
     \begin{subfigure}[t]{0.45\textwidth}
         \centering
         \includegraphics[width=\textwidth]{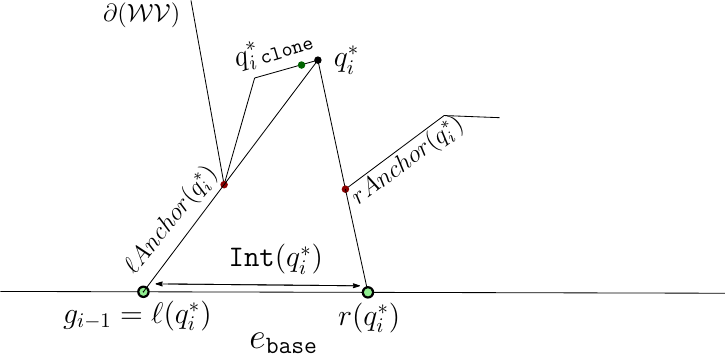}
         \subcaption{\textbf{Subcase} \hyperref[ProofWitnessCase1]{2a}: $q_i^*$ is a vertex of $\wv$.}
         \label{figcaseA1}
     \end{subfigure}
     \hfill
     \begin{subfigure}[t]{0.5\textwidth}
         \centering
         \includegraphics[width=\textwidth]{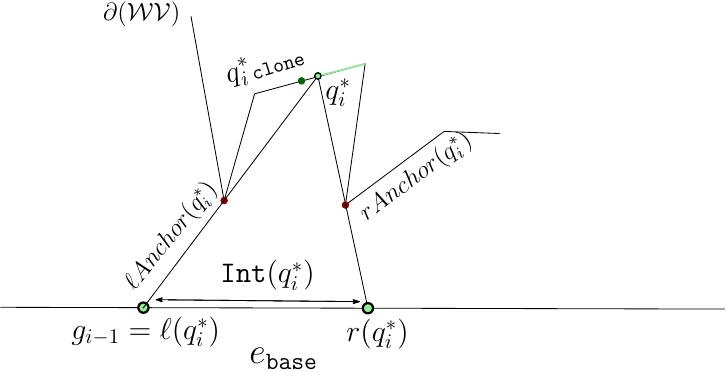}
         \subcaption{\textbf{Subcase} \hyperref[ProofWitnessCase2]{2b}: $q_i^*$ lies on the interior of some edge $e$.}
         \label{figcaseA2}
     \end{subfigure}
        \caption{Producing Witness candidates using Clone Points}
        \label{inducedvertex}
\end{figure}

\medskip
\noindent\textbf{Step 4: $W$ is pairwise separated.}

Fix $j<i$. By the Claim in Step 1, we only need to worry about the case where $\I(q_j^*)$ and $\I(q_i^*)$ touch, i.e. $j=i-1$ with $q_i^*$ a Case 2 point; everywhere else $\I(q_j^*)$ and $\I(q_i^*)$ have a positive gap, which survives however $w_i,w_j$ are chosen, by Step 2. We consider this in the four possible exhaustive cases.

\smallskip
\noindent\textit{(q, q).} $w_i = q_i^*$, $w_j = q_j^*$. Since $w_i = q_i^*$, the point $q_i^*$ is a Case 1 point, so by the Claim, $\I(q_j^*) \cap \I(q_i^*) = \emptyset$ for every $j<i$. Hence $\I(w_i) \cap \I(w_j) = \emptyset$ directly.

\smallskip
\noindent\textit{(q, clone).} $w_i = q_i^*$, $w_j = \clv[j]{q}$. Again $q_i^*$ is Case 1, so $\I(q_j^*) \cap \I(q_i^*) = \emptyset$ with a positive gap, by the Claim. By Step 2, we choose $\clv[j]{q}$ close enough to $q_j^*$ that this gap survives, giving $\I(w_j) \cap \I(w_i) = \emptyset$.

\smallskip
\noindent\textit{(clone, q).} $w_i = \clv[i]{q}$, $w_j = q_j^*$.

If $j < i-1$: by the Claim, $\I(q_j^*) \cap \I(q_i^*) = \emptyset$ with a positive gap, and by Step 2 we choose $\clv[i]{q}$ close enough to $q_i^*$, on top of satisfying $(\ast\ast)$, that this gap survives.

If $j = i-1$: since $w_j = q_j^*$ is the actual point, $r(w_j) = r(q_{i-1}^*) = g_{i-1}$ exactly. By $(\ast\ast)$, $\ell(w_i) = \ell(\clv[i]{q}) > g_{i-1} = r(w_j)$. So $\I(w_i) \cap \I(w_j) = \emptyset$ by construction, with no further adjustment needed. This is the situation in \cref{figcaseA1,figcaseA2}.

\smallskip
\noindent\textit{(clone, clone).} $w_i = \clv[i]{q}$, $w_j = \clv[j]{q}$.

If $j<i-1$: as above, $\I(q_j^*)\cap\I(q_i^*)=\emptyset$ with a positive gap, and by Step 2, we choose both clones close enough to their source points to preserve them.

If $j = i-1$: this is the case not yet handled by the construction in Step 3, since there $q_i^*$ was only compared against actual points, not against a clone of $q_{i-1}^*$. We resolve it now. By $(\ast\ast)$, placing $\clv[i]{q}$ opens a gap $\varepsilon_i = \ell(\clv[i]{q}) - g_{i-1} > 0$ above $g_{i-1}$. Since $\clv[i-1]{q}$ is a clone of $q_{i-1}^*$, we have $r(\clv[i-1]{q})$ moves to $r(q_{i-1}^*) = g_{i-1}$ as $\clv[i-1]{q}$ moves to $q_{i-1}^*$, by continuity of $r(\cdot)$ near $q_{i-1}^*$. So we may choose $\clv[i-1]{q}$ close enough to $q_{i-1}^*$ that
\[
    r(\clv[i-1]{q}) \;<\; g_{i-1} + \varepsilon_i \;=\; \ell(\clv[i]{q}).
\]
This gives $\I(\clv[i-1]{q}) \cap \I(\clv[i]{q}) = \emptyset$, as required.

It remains to check that shrinking $\clv[i-1]{q}$ this way does not undo any earlier requirement already placed on it, for instance that it avoids $g_{i-2}$ if $q_{i-1}^*$ was itself a Case 2 point. Every such earlier requirement was itself of the form "close enough to $q_{i-1}^*$", so moving $\clv[i-1]{q}$ even closer to $q_{i-1}^*$ might change the non-intersecting property with the earlier interval, $\I(q_{i-2}^*)$. In that case, we again readjust our choice of clone for the witness $w_{i-2}$, just in the same way we did for $w_{i-1}$. It may happen that for some particular $w_i$, we will have to readjust our clone point choices for every preceding $j$, where $j<i$. However, existentially we always have a choice for a witness point for all $i = 1, 2, \cdots, k$. Hence, we can always ensure to have $\I(w_j) \cap \I(w_i) = \emptyset$.

\medskip
\noindent\textbf{Conclusion.} These four cases cover every pair $w_i,w_j \in W$, and in each we found witnesses with $\I(w_i)\cap\I(w_j)=\emptyset$. So $W$ is a witness set of size $k$. Hence, any guard set for $\wv$ with respect to $\eb$ has size at least $k$, since the intervals of the elements of $W$ are pairwise disjoint, any valid guard set must contain at least one guard lying within each of these $k$ mutually disjoint intervals.  Therefore, $\mathrm{OPT} \geq k$.
\end{proof}


\subsection{Part II: Completeness (Every Boundary Point is Guarded)}

\begin{lemma}\label{lemma:complete}
At termination (upto $k$-steps or till $\mathcal{U}_i = \emptyset$), every point $p \in \bd(\wv)$ satisfies $\I(p) \cap G_{\OPT} \neq \emptyset$.
\end{lemma}

\begin{proof}
After running the above algorithm for $k$ steps, if $\mathcal{U}\neq\emptyset$, then there exists $p \in Q$ such that $\I(p)$ doesn't intersect with any guard from $G_{k}$, and hence we can add $p$ to the witness set $W$ and get another witness candidate. So the total number of witness points is at least $k+1$. But as $\grN(\wv,\eb) \geq \witN(\wv,\eb)$, it implies that the minimum number of guards needed to guard the polygon is at least $k+1$. Hence, the weak visibility polygon $\wv$ is not $k$-guardable.

Hence, assume, after k iterations, $\mathcal{U}=\emptyset$. We show that via our algorithm, all points of $\bd(\wv)$ have been guarded. To prove that, we argue by contradiction. Suppose there exists $p^* \in \bd(\wv)$ such that $\I(p^*) \cap G_{k} = \emptyset$, i.e., $p^*$ is unguarded at termination.
 
\medskip
\noindent\textbf{Case 1: $p^*$ is a vertex.}
Then $p^*$ would have been included in $Q$ in the very first step, so $\I(p^*)$ is present from the very first iteration. Suppose the algorithm runs for $\le k$ iterations before halting with $\mathcal{U}=\emptyset$. If $\I(p^*)$ were never covered, then $p^*$ would still be in $\mathcal{U}$ at termination, contradicting $\mathcal{U}=\emptyset$.
 
\medskip
\noindent\textbf{Case 2: $p^*$ is a non-vertex boundary point.}

Let $e = (u, v)$ be the edge of $\bd(\wv)$ containing $p^*$, with $u$ the left endpoint and $v$ the right endpoint along the fixed orientation of $\bd(\wv)$.  Let $Q\!\restriction_e$ denote the set of candidate points in $Q$ that lie on $e$.  This set is non-empty because $u,v\in Q$ from initialization. Order the points of $Q\!\restriction_e$ along $e$ from left to right.  Let $u^*$ be the point in $Q\!\restriction_e$ lying immediately to the left of $p^*$, and $v^*$ be the point in $Q\!\restriction_e$ lying immediately to the right of $p^*$ (with $u^*=u$ if no expansion point was added between $u$ and
$p^*$, and similarly $v^*=v$).  By construction, no point of $Q$ lies strictly between $u^*$ and $v^*$ on $e$.
 
Since $u^*, v^* \in Q$, their intervals $\I(u^*)$ and $\I(v^*)$ were placed into $\mathcal{U}$ at some iteration. The algorithm terminates only when $\mathcal{U} = \emptyset$, so both $u^*$ and $v^*$ are guarded at termination. Let $g_{u^*}$ and $g_{v^*}$ be guards in $G_{k}$ covering $u^*$ and $v^*$, respectively.
 
\begin{claim}\label{claim:consecutive}
There exist guards $g_L,g_R\in G_{k}$ such that $g_L\in\I(v^*)$, $g_R\in\I(u^*)$, and either they are same (i.e., $g_L = g_R$) or $g_L,g_R$ are \emph{consecutive} in $G_{k}$ along $\eb$, meaning no guard of $G_{k}$ lies strictly between $g_L$ and $g_R$.
\end{claim}
 
\begin{proof}[Proof of Claim~\ref{claim:consecutive}]
Define
\[
  g_L \;:=\; \max_x\bigl\{g\in G_{k} : g \in \I(v^*)\bigr\},
  \qquad
  g_R \;:=\; \min_x\bigl\{g\in G_{k} : g \in \I(u^*)\bigr\}.
\]
(These are well-defined since both $\I(u^*)$ and $\I(v^*)$ contain at least one guard). We can evaluate their positions by looking at two possibilities:

\begin{itemize}
    \item \textbf{Case 1: $g_L =_x g_R$} \\
    A single guard $g\in G_{k}$ lies in $\I(u^*)\cap\I(v^*)$. Since $u^*$ and $v^*$ are on the same edge $e$ with no reflex vertex between them, the strong visibility interval of the open sub-edge $(u^*,v^*)$ satisfies
    \[
      \SI(u^*v^*) \;\supseteq\; \I(u^*)\cap\I(v^*),
    \]
    and in particular $g\in\SI(u^*v^*)$.  Hence $g$ sees every point of $(u^*,v^*)$, including $p^*$, contradicting $\I(p^*)\cap G_{k}=\emptyset$.

    \item \textbf{Case 2: $g_L >_x g_R$} \\
    We claim that this case cannot occur. Suppose for contradiction $g_L > g_R$. If $g_L > g_R$, then $g_L\in\I(v^*)$ and $g_R\in\I(u^*)$ with $g_R<g_L$.  Then by definition, $g_R\in\I(u^*)$ means the line segment $g_R\,u^*$ is contained in $\wv$. Since $u^*\,v^*$ is part of the edge, $g_R$ also sees $v^*$ (otherwise if the line segment $g_R\,v^*$ is not inside $\wv$, then there has to be an obstruction, which will also obstruct the visibility of $v^*$ from $g_L$). Hence $g_R \in \I(v^*)$, a contradiction to the definition of $g_L$, as $g_L$ is the leftmost guard guarding $v^*$.

    \item \textbf{Case 3: $g_L <_x g_R$} \\
    Suppose some $g\in G_{k}$ satisfies $g_L < g < g_R$. By the greedy placement rule (guards are placed at right endpoints of uncovered intervals, processed by increasing right endpoint), $g$ was placed at some iteration to cover an interval $\I(q)$ for some $q\in Q$.  Since $g_L<g<g_R$, guard $g$ sees some candidate witness $q\in Q$.  This witness $q$ must lie on edge $e$ between $u^*$ and $v^*$: if $q$ lay outside $(u^*,v^*)$, then by the ordering of $Q\!\restriction_e$ either $q\leq u^*$ or $q\geq v^*$ on $e$, but $g\in(g_L,g_R)$ and the visibility interval of any point left of $u^*$ (resp.\ right of $v^*$) has right endpoint $\leq r(u^*)< g_R$ (resp.\ left endpoint $\geq \ell(v^*)>g_L$), so $g$ could only cover such a $q$ if $g\leq r(u^*) < g_R$ or $g\geq\ell(v^*) > g_L$---we can say, any point of $Q$ outside $(u^*,v^*)$ whose interval contains $g$ would be a closer neighbour of $p^*$ in $Q\!\restriction_e$ than either $u^*$ or $v^*$, contradicting the choice of $u^*$ and $v^*$ as immediate neighbours. 

    Therefore $q\in(u^*,v^*)\cap Q$, contradicting the assumption that $u^*$ and $v^*$ are immediate neighbours of $p^*$ in $Q$.  Hence, no such $g$ exists, and $g_L,g_R$ are consecutive (see \cref{fig:guarding-completeness}). \qedhere
\end{itemize}
\end{proof}

\begin{figure}[H]
    \centering
    \includegraphics[width=0.75\linewidth]{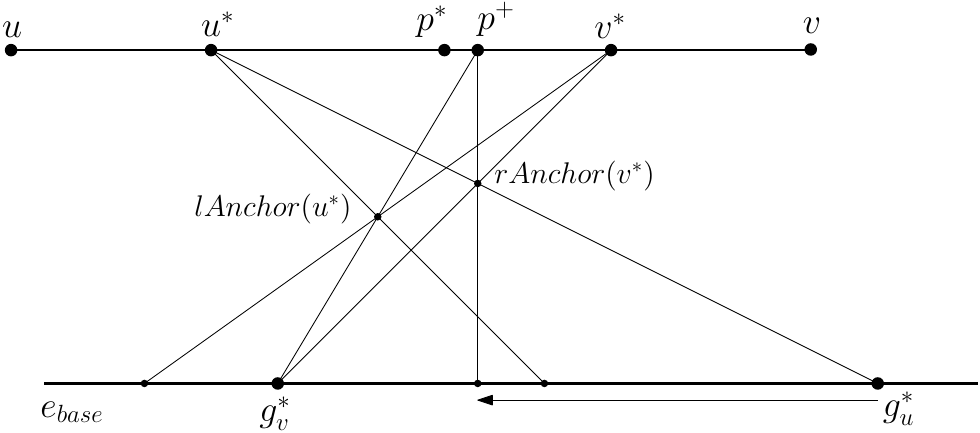}
    \caption{Illustration of the Proof of Case 2: guards $g_{u^*}^*$ and $g_{v^*}^*$ with anchor rays through $\LA(u^*)$ and $\RA(v^*)$ locating the witness point $p^+$.}
    \label{fig:guarding-completeness}
\end{figure}
 
Let $g_{v^*}$ and $g_{u^*}$ denote the two consecutive guards from Claim~\ref{claim:consecutive}, with $g_{v^*} < g_{u^*}$ on $\eb$ (the case $g_{u^*} = g_{v^*}$ was handled above) such that $g_{v^*}$ sees $v^*$ and $g_{u^*}$ sees $u^*$.
 
Now join $g_{v^*}$ with $\LA(u^*)$ of $u^*$, and let this extended line meet the edge $e_{u,v}$ at a point $p^+$ beyond $\LA(u^*)$. By the definition of the left anchor and the interval map, $p^+$ is the boundary point whose left interval endpoint $\ell(p^+)$ equals $g_{v^*}$; that is, $\ell(p^+) = g_{v^*}$
 
Since $p^+$ lies on $e_{u,v}$ strictly to the right of $u^*$ and to the left of $v^*$ (by the consecutiveness of $g_{u^*}$ and $g_{v^*}$ and the visibility geometry), and $g_{u^*}$ does \emph{not} see $p^+$ (as $g_{u^*}$ is to the right of $r(p^+)$ by the anchor ray construction), we have: $r(p^+) < g_{u^*}$
 
By the witness expansion step, $p^+$ is added to $Q$ when the guard $g_{v^*}$ is placed and the ray through $\LA(u^*)$ is shot. Hence $\I(p^+)$ enters $\mathcal{U}$ in the next iteration. The algorithm selects guards by their leftmost right endpoints, so when $\I(p^+)$ is uncovered, it places a guard at $r(p^+)$. But $r(p^+) < g_{u^*}$ means the algorithm would have placed $g_{u^*}$ at $r(p^+)$ rather than at its actual position---a contradiction with the selection rule. 
 
Hence, no such unguarded $p^*$ exists, and every non-vertex boundary point is covered at termination. 
\end{proof}


\section{Runtime Analysis of the Algorithm}\label{sec:algoruntime}

The algorithm outputs a minimum guard set $G \subseteq \eb$ and a \emph{candidate witness set} $\widehat{W} = \{\widehat{w}_1, \dots, \widehat{w}_k\}$. The points in $\widehat{W}$ are not themselves the witnesses: each $\widehat{w}_i$ is either a vertex of $\wv$ or a boundary intersection point produced by the ray-shooting step. We show that for each $\widehat{w}_i$ there exists a \emph{clone point} $w_i$ in its vicinity that serves as a valid witness, i.e.\ $w_i$ is visible from $g_i$ but not from any other guard in $G$. The existence of these clone points is established in \cref{lemma:opt}, ensuring that $W = \{w_1, \dots, w_k\}$ is a valid witness set of size $k = \OPT$.

\subsection{Runtime Analysis}

Let $n = |V(\wv)|$, let $\rho$ denote the number of reflex vertices of $\wv$, and recall
that the algorithm places $k = \OPT$ guards. Note that $\rho \le n$, and that $\rho$ may be
much smaller than $n$; we therefore keep $\rho$ as an explicit parameter. We show that the
algorithm runs in $\OO\bigl((n + k\rho)(\log n + \log k)\bigr)$ time. The two logarithmic
factors have distinct origins, and we keep them apart throughout: $\log n$ comes from
queries against the fixed polygon, and $\log k$ from operations on data structures whose
size grows with the number of guards placed. The parameter $\rho$ enters only through the
\emph{number} of such operations, never through the cost of an individual one.

\subsubsection{Preprocessing.}
We build the shortest-path tree (SPT) \cite{GUIBAS1989126} of $\wv$ rooted at the left
endpoint of $\eb$ and the SPT rooted at the right endpoint, each in $\OO(n)$ time. From
these two trees we compute, in $\OO(n)$ total time, the visibility interval
$\I(v) \subseteq \eb$ of every vertex $v \in V(\wv)$; the $\rho$ reflex vertices are
identified during the same pass. For the ray-shooting queries, we use the data structure of
Hershberger and Suri \cite{DBLP:journals/jal/HershbergerS95}, built once on $\wv$ in
$\OO(n)$ time and $\OO(n)$ space, answering each query in $\OO(\log n)$ time. We also
initialize two dynamic structures used by the main loop: a min-heap $H$ that stores
candidate intervals keyed by their right endpoints $r(\cdot)$, and a balanced search tree
$T$ that stores the guards placed so far as points of $\eb$. Loading the $n$ initial intervals into $H$ takes $\OO(n)$ time by bottom-up heap
construction; this batch bound applies only here, since during the main loop
insertions and extractions are forcibly interleaved and are charged as individual
heap operations. Preprocessing, therefore, costs $\OO(n)$.

\subsubsection{Total number of rays.}
It is tempting to bound the number of boundary points created over the whole execution by
$\rho$, on the grounds that each reflex vertex produces one such point. This is not
correct: a reflex vertex can be visible from several guards, and the loop shoots a fresh
ray through it once for each such guard, creating a distinct boundary point every time.
The correct bound comes from charging each ray to the pair (guard, reflex vertex) that
fires it.

\begin{lemma}\label{lem:raycount}
Let $r_i$ denote the number of reflex vertices of $\wv$ visible from the guard $g_i$. Then
\[
  \sum_{i=1}^{k} r_i \;\le\; k\rho .
\]
\end{lemma}
\begin{proof}
In iteration $i$, the loop shoots exactly one ray per reflex vertex visible from $g_i$, so
a reflex vertex is charged in iteration $i$ if and only if it is visible from $g_i$. Each
of the $\rho$ reflex vertices is therefore charged at most once in each of the $k$
iterations.
\end{proof}

By \cref{lem:raycount} the candidate set $Q$ grows to size at most $n + k\rho$, and this
is tight in the worst case. Consequently the heap $H$ holds at most $n + k\rho$ intervals at any time, so a
single heap operation costs $\OO\bigl(\log(n + k\rho)\bigr) = \OO(\log n + \log k)$, using $\rho \le n$ (The $\OO(n)$ bottom-up construction cannot be reused inside the loop: the intervals created in iteration $i$ do not exist until the extractions of that iteration have fixed $g_i$, so the operations cannot be batched). Any
implementation that re-sorts or linearly scans the entire candidate set in each iteration
would spend $\Omega(n + i\rho)$ time in iteration $i$, hence
$\Omega\bigl(kn + k^{2}\rho\bigr)$ overall; the heap avoids the quadratic dependence on
$k$, as we show next.

\subsubsection{Main loop.}
The loop runs for $k$ guard-placing iterations, followed by a final iteration that
certifies $\mathcal{U}_i = \emptyset$. We bound the two steps of an iteration separately
and then sum.

\emph{Selecting the next guard.} To choose $g_i$ we need the eligible candidate interval
with the smallest right endpoint, where $\I(q)$ is eligible when $\I(q)\cap G_{i-1}$ is
empty, or equal to $\{g_{i-1}\}$ with $\ell(q) = g_{i-1} < r(q)$. This is a minimum, not a
full order, so a heap suffices. We repeatedly extract the minimum of $H$ and test the
extracted interval for eligibility against $T$. Once $\I(q)$ contains a placed guard that
is not the most recent one, it stays ineligible in every later iteration, since guards are
only ever added and such a guard can never again be the most recent; we therefore discard
such an interval and extract again. The first eligible interval $\I(q^{*})$ we reach fixes
$g_i = r(\I(q^{*}))$, and we insert $g_i$ into $T$. An extraction costs
$\OO(\log n + \log k)$ as noted above; an eligibility test uses a constant number of rank
queries in $T$, which holds at most $k$ guards, and costs $\OO(\log k)$. The final
iteration certifies emptiness by draining the heap through the same
extract-and-test loop.

We do not bound the number of extractions per iteration, because a single guard may render
many intervals ineligible at once. Instead, we charge globally: every interval enters $H$
exactly once and leaves at most once, so the extractions and eligibility tests summed over
all iterations, including the final drain, number at most the insertions, namely
$|Q_0| + \sum_{i} r_i \le n + k\rho$ by \cref{lem:raycount}. Selection therefore costs
$\OO\bigl((n + k\rho)(\log n + \log k)\bigr)$ in total.

\emph{Ray shooting and expansion.} For each of the $r_i \le \rho$ reflex vertices visible
from $g_i$ we shoot a ray from $g_i$ through the vertex and intersect it with $\bd(\wv)$;
each query runs against the fixed Hershberger--Suri structure on the $n$-vertex polygon in
$\OO(\log n)$ time, so ray shooting costs $\OO(r_i\log n)$ for the iteration
\cite{DBLP:journals/jal/HershbergerS95}. We stress that this $\log n$ does not improve to
$\log\rho$ when $\rho$ is small: the parameter $\rho$ bounds how many queries are made,
while the cost of one query is governed by the size of the polygon. Each of the $r_i$ new
boundary points $b_j$ then receives its visibility interval $\I(b_j)$ in $\OO(1)$ time
from the two precomputed SPTs, and is inserted into $H$ in $\OO(\log n + \log k)$ time.
The iteration therefore costs $\OO\bigl(r_i(\log n + \log k)\bigr)$, and by
\cref{lem:raycount} this step costs
$\sum_{i} \OO\bigl(r_i(\log n + \log k)\bigr) = \OO\bigl(k\rho(\log n + \log k)\bigr)$
over the whole loop.

\subsubsection{Total runtime.}\label{subsection:runtime}
In total, $n + k\rho$ intervals are ever created. Each is touched by a constant number of
times: one interval computation in $\OO(1)$, one heap insertion, at most one extraction,
and one eligibility test, each in $\OO(\log n + \log k)$; the $k\rho$ expansion intervals
additionally cost one ray-shooting query each, in $\OO(\log n)$. Adding the $\OO(n)$
preprocessing, the algorithm runs in
\[
  \OO\bigl((n + k\rho)\,(\log n + \log k)\bigr)
  \;=\; \OO\bigl((n + \OPT\cdot\rho)\,(\log n + \log\OPT)\bigr)
\]
time. Two specializations are worth recording. When $\rho = \Theta(n)$ the bound reads
$\OO\bigl(\OPT\cdot n(\log n + \log\OPT)\bigr)$, matching the analysis without the
parameter $\rho$. When $\rho = \OO(1)$ the bound collapses to
$\OO\bigl((n + k)(\log n + \log k)\bigr)$; if in addition $k \le n$, this is
$\OO(n\log n)$, matching the $\Omega(n\log n)$ lower bound of
\cref{thm:gs-lower-bound}. Thus, the algorithm is optimal on polygons with
a constant number of reflex vertices, and the parameter $\rho$ isolates the expansion step as
the sole source of superlinear cost.

When $\rho = \OO(1)$ the bound collapses to $\OO\bigl((n + k)(\log n + \log k)\bigr)$,
which is near-linear in the input size plus the output size; note that $k$ is not
bounded by any function of $n$ even in this regime (\cref{prop:single-edge-unbounded}), so
the dependence on $k$ cannot be removed. If moreover $k \le n$, the bound is
$\OO(n\log n)$, matching the lower bound of \cref{thm:gs-lower-bound}.

\subsubsection{Model of computation.}
The bound above counts operations in the real-RAM model, in which each arithmetic operation on a real coordinate costs $\OO(1)$. Two independent sources of cost deserve separate mention. First, the factor $\log\OPT$ is combinatorial: it arises from the sizes of the heap $H$ and the tree $T$, which grow with the number of guards, and it is present already in the real-RAM model. Second, the coordinates of the guards and of the expansion points $b_j$ are produced by intersecting a line through a previously constructed point with an edge of $\bd(\wv)$; iterating this construction composes linear-fractional maps, so the algebraic degree and the bit-length of these coordinates can grow with the number of iterations. In a bit-complexity model, the running time incurs an additional factor that is polynomial in the coordinate bit length. We state the model explicitly because computing the exact optimum of art-gallery guarding is $\exists\mathbb{R}$-hard \cite{DBLP:journals/jacm/AbrahamsenAM22}: a hardness about the algebraic complexity of optimal coordinates, not about the number of guards (recall $\exists\mathbb{R}\subseteq\mathrm{PSPACE}$). The coordinate axis and the combinatorial $\log \OPT$ factor are therefore distinct, and neither subsumes the other.

\medskip

So far, we have obtained the following.
The \cref{alg:base-guard} produces a minimum guard set $G \subseteq \eb$ and a candidate witness set $\widehat{W} \subseteq \wv$ such that:
\begin{enumerate}
  \item By \cref{lemma:complete}, $|G| = |\widehat{W}| = \OPT$, the minimum number of guards on $\eb$ needed to guard $\bd(\wv)$ (and hence all of $\wv$, Theorem \ref{thm:boundary-equivalence}).
  \item For each candidate witness $\widehat{w}_i \in \widehat{W}$, there exists a clone point $w_i$ in the vicinity of $\widehat{w}_i$ that is a valid witness for guard $g_i$; the existence of these clone points is established in \cref{lemma:opt}.
  \item The resulting witness set $W = \{w_1, \dots, w_{\OPT}\}$ is a maximum independent set in the interval graph $\mathcal{G}_{Q_{\OPT}}$, and $G$ forms a minimum clique cover of $\mathcal{G}_{Q_{\OPT}}$, certifying optimality via guard--witness duality.
  \item The algorithm runs in $\OO\bigl((n + \OPT\cdot\rho)\,(\log n + \log\OPT)\bigr)$ time, where $n = |V(\wv)|$, and $\rho$ is the number of reflex vertices in $\wv$ (\ref{subsection:runtime}).
\end{enumerate}

Thus, we have proved the main result of our paper.

\MainResult*


\section{Guarding the vertices of a \texorpdfstring{$\wvp$}{wvp} from base}\label{sec:vertexguardingWV}

In this section, we look at the guarding problem for vertices in a $\wvp$, where the guards are located at its base edge. We claim that in this case, the optimal guard set can be computed in $\OO(n \log n)$ time. To prove the theorem, we will need a lemma, which is stated below.

\begin{lemma}\label{lemma:MinCliqueGuardSet}
    Let $\mathcal{I} = \{\I(v) : v \in V(\wv)\}$ be the set of intervals on $\eb$ defined for each vertex of $\wv$, and let $\mathcal{G}_Q = (V(\wv), E)$ be the interval graph induced by $\mathcal{I}$, where $\{u,v\} \in E$ if and only if $\I(u) \cap \I(v) \neq \emptyset$. Then the minimum number of guards on $\eb$ required to see all vertices of $V(\wv)$ equals the size of a minimum clique cover of $\mathcal{G}_Q$.
\end{lemma}
\begin{proof}
    We prove both directions of the equivalence.

    \medskip
    \noindent\textit{(Guard set $\Rightarrow$ Clique cover.)} Let $G = \{g_1, g_2, \ldots, g_k\} \subseteq \eb$ be a guard set for $V(\wv)$. For each guard $g_i$, define the set $C_i = \{v \in V(\wv) : g_i \in \I(v)\}$ of all vertices seen by $g_i$. We claim each $C_i$ is a clique in $\mathcal{G}_Q$. Indeed, for any two vertices $u, w \in C_i$, we have $g_i \in \I(u)$ and $g_i \in \I(w)$ by definition, and therefore $g_i \in \I(u) \cap \I(w)$, so $\I(u) \cap \I(w) \neq \emptyset$, meaning $\{u,w\} \in E$. 
    Since $G$ is a guard set for $V(\wv)$, every vertex $v \in V(\wv)$ is seen by some guard $g_i$, so $\bigcup_{i=1}^k C_i = V(\wv)$. Therefore $\{C_1, \ldots, C_k\}$ is a clique cover of $\mathcal{G}_Q$ of size $k$.

    \medskip
    \noindent\textit{(Clique cover $\Rightarrow$ Guard set.)} Let $\{C_1, C_2, \ldots, C_k\}$ be a clique cover of $\mathcal{G}_Q$. We claim that for each clique $C_i$, the intersection $\bigcap_{v \in C_i} \I(v)$ is non-empty. Since $\{C_i\}$ is a clique, every two intervals in $\{\I(v) : v \in C_i\}$ intersect. Intervals on a line satisfy the Helly property \cite{deberg2008computational, DanzerGruenbaumKlee1963}: a finite family of intervals on a line has a common point if and only if every two of them intersect. Therefore $\bigcap_{v \in C_i} \I(v) \neq \emptyset$, and we may place a guard $g_i \in \bigcap_{v \in C_i} \I(v)$. By construction, $g_i \in \I(v)$ for every $v \in C_i$, so $g_i$ sees all vertices of $C_i$. Since $\{C_1, \ldots, C_k\}$ covers $V(\wv)$, the set $\{g_1, \ldots, g_k\}$ is a valid guard set of size $k$.

    \medskip
    Together, the two directions show that any guard set of size $k$ yields a clique cover of size $k$, and any clique cover of size $k$ yields a guard set of size $k$. 
    Therefore, the minimum sizes of both objects coincide, establishing the equivalence.
\end{proof}

\begin{lemma}\label{lemma:vertex-guard-upper}
    For a $\wvp ~ \wv$, $\gs(V(\wv), \eb)$ can be solved in $\OO(n\log n)$ time.
\end{lemma}
\begin{proof}
    We begin by computing the visibility region of each vertex $v$ of $\wv$ $(= \vis(v))$, and determining its intersection with the base edge $\eb$. Precisely, we find the interval $\I(v)$ for every vertex $v \in V(\wv)$.

    By \cref{prop:vis-equiv}, a guard $g \in \eb$ sees a vertex $v$ if and only if $g \in \I(v)$. Consequently, since guard locations are restricted to $\eb$, the problem of finding $\gs(V(\wv), \eb)$ reduces to a covering problem on these intervals. 
    Specifically, the intervals $\{\I(v) : v \in V(\wv)\}$ form an interval graph \cite{GOLUMBIC20041}, wherein for two vertices $u, v$ of $\wv$, the condition $\I(u) \cap \I(v) = \emptyset$ implies that two guards are required to guard $u$ and $v$, while $\I(u) \cap \I(v) \neq \emptyset$ implies that a single guard placed anywhere in $(\I(u) \cap \I(v)) \subseteq \eb$ suffices to guard both $u$ and $v$. 
    Therefore, finding $\gs(V(\wv), \eb)$ is equivalent to finding a minimal clique cover of this interval graph, due to \cref{lemma:MinCliqueGuardSet}. 

    To construct the interval graph, it suffices to compute, for each vertex $v$, the two anchor points $\RA(v)$ and $\LA(v)$, rather than its entire visibility region. These anchor points arise naturally during the computation of shortest paths from $v$ to the two endpoints of $\eb$. The interval $\I(v)$ is then obtained by extending the line segments $\overline{v\,\RA(v)}$ and $\overline{v\,\LA(v)}$ until they intersect $\eb$ (\cref{def-interval}).

    Since $\eb = (u, v)$, for each vertex $v_i$, the anchor points $\LA(v_i)$ and $\RA(v_i)$ are found by constructing shortest path trees rooted at $u$ and $v$, respectively. Following \cite{GUIBAS1989126}, both trees can be constructed in linear time. In each shortest path tree, the vertices of $\wv$ appear as leaf nodes, and the required anchor point for a vertex $v_i$ is precisely its predecessor in the corresponding shortest path tree.

    The construction of all intervals therefore takes $\OO(n)$ time, and the minimal clique cover of the resulting interval graph can be computed in $\OO(n \log n)$ time \cite{https://doi.org/10.1002/net.3230120410}. Combining these bounds, the total time required to solve $\gs(V(\wv), \eb)$ is $\OO(n \log n)$.
\end{proof}

\medskip

Having established an $\OO(n\log n)$ upper bound for $\gs(V(\wv), \eb)$, we now show this bound is tight. We do so by reducing the problem of computing the size of a minimum clique cover (MCC) of an interval graph to $\gs(V(\wv),\eb)$, together with the following known lower bound on that problem.

\begin{lemma}\label{lemma:mcc-lower-bound}
    Given $n$ intervals on the real line, computing the size of a minimum clique cover of the interval graph they induce requires $\Omega(n\log n)$ time in the algebraic decision-tree model. This follows from a standard reduction from the \textsc{Sorting} (equivalently, \textsc{Element Uniqueness}) problem, which itself requires $\Omega(n\log n)$ time in this model~\cite{DBLP:conf/stoc/Ben-Or83, DBLP:books/sp/PreparataS85}.
\end{lemma}

Now, we show that any algorithm that solves $\gs(V(\wv), \eb)$ requires $\Omega(n\log n)$ time in the worst case (\cref{lemma:mcc-lower-bound}). Combining this with the previous \cref{lemma:vertex-guard-upper}, $\gs(V(\wv), \eb)$ can be solved in $\Theta(n \log n)$ time.

\vsg*

\begin{proof}
    Suppose, for contradiction, that there exists an algorithm $\mathcal{A}$ that solves $\gs(V(\wv),\allowbreak\, \eb)$, for any $\wvp ~ \wv$ on $n$ vertices, in ${o}(n\log n)$ time. We use $\mathcal{A}$ to compute the size of a minimum clique cover of an arbitrary interval graph in $o(n\log n)$ time, contradicting \cref{lemma:mcc-lower-bound}.

    \medskip
    \noindent\textit{Reduction.} Let $I_1, \ldots, I_n$ be $n$ intervals on the real line, $I_j = [l_j, r_j]$, inducing an interval graph $\mathcal{G}_Q$. Let $x = \min_j l_j$, $y = \max_j r_j$, and set $a = x-1$ and $b = y+1$. Let $\eb$ be the segment joining $(a,0)$ and $(b,b-a)$; since this segment has slope $1$, every point on it can be written uniquely as $(s, s-a)$ for $s \in [a,b]$, which we call its \emph{coordinate}. As $[x,y] \subseteq [a,b]$, each $I_j$ corresponds to a well-defined sub-segment of $\eb$.

    Fix $n$ points $t_1 < \cdots < t_n$ in $(a,b)$, spaced well apart from one another (see \cref{fig:omega-nlogn}). We build a $\wvp ~ \wv$ with base edge $\eb$ whose opposite boundary chain consists of $n$ shallow spikes, one centered above each $t_i$. The apex of the $i$-th spike is a vertex $p_i$ whose two incident edges are oriented so that, extended as full lines, they meet $\eb$ exactly at coordinates $l_i$ and $r_i$. Since each spike can be made arbitrarily narrow and shallow, and the $t_i$ are well separated, the two slopes at $p_i$ can always be chosen to realize any prescribed pair $l_i, r_i \in [a,b]$, independently of every other spike and without one spike obstructing another's view of $\eb$. By the anchor-point characterization of $\I(\cdot)$ used in the proof above, this gives $\I(p_i) = [l_i, r_i] = I_i$ exactly. Each spike is specified in $\OO(1)$ time, so $\wv$ is built in $\OO(n)$ time, and the interval graph induced on $V(\wv) = \{p_1,\ldots,p_n\}$ by $\{\I(p_i)\}_{i=1}^n$ is precisely $\mathcal{G}_Q$.

    \begin{figure}[H]
        \centering
        \includegraphics[width=0.6\linewidth]{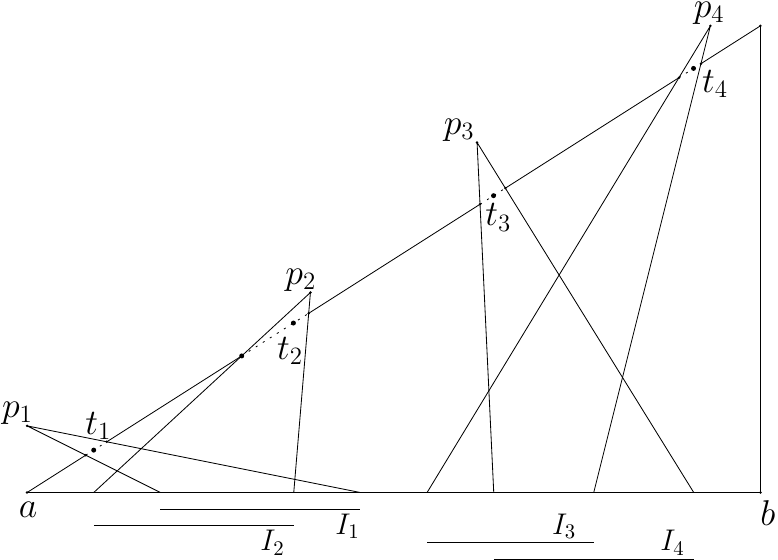}
        \caption{Illustrative figure. The spike at $p_i$, centered above $t_i$, has its edges extended to meet $\eb$ at coordinates $l_i, r_i$, so that $\I(p_i) = I_i$.}
        \label{fig:omega-nlogn}
    \end{figure}

    \medskip
    \noindent\textit{Contradiction.} Run $\mathcal{A}$ on $\wv$ to obtain a minimum guard set $G$ for $V(\wv)$ on $\eb$, in $o(n\log n)$ time. By \cref{lemma:MinCliqueGuardSet}, $|G|$ equals the minimum clique cover number of the interval graph induced by $\{\I(p_i)\}_{i=1}^n = \{I_i\}_{i=1}^n$, i.e.\ of $\mathcal{G}_Q$ itself and this size is obtained simply by counting the guards $\mathcal A$ returns, in $\OO(1)$ further time. So we compute the MCC size of $\mathcal{G}_Q$ in $\OO(n) + o(n\log n) = o(n\log n)$ time, contradicting \cref{lemma:mcc-lower-bound}. Hence no such $\mathcal{A}$ exists, and $\gs(V(\wv),\eb)$ requires $\Omega(n\log n)$ time.
\end{proof}


\section{Perfect Guarding of an Altitude Terrain at a Fixed Height} \label{sec:fixed-height}

Authors in \cite{DAESCU201922} place an optimal number of guards on a fixed altitude line above a monotone mountain in $\OO(n)$ time, and they prove that a monotone mountain is \emph{perfect}: its guard number equals its witness number. Certifying that perfectness, however, means exhibiting an actual witness set of matching size, and the previously best bound for producing one is $\OO(n^2\log n)$. We compute, at a fixed altitude line, an optimal guard set \emph{and} a matching optimal witness set together in $\OO(n)$ time, so that the output certifies its own optimality.

Throughout, $\mo$ is an $x$-monotone mountain, $L$ is a fixed altitude line at height $y$ lying above $\mo$, and $k$ is the minimum number of guards needed on $L$. For a terrain point $w$ that is visible from some point of $L$, we write $\I_y(w) = [\ell(w,y),\, r(w,y)]$ for its visibility interval on $L$; a guard that sees $w$ from the right end of that interval sits at $r(w,y)$.

\subsection{Two Passes of linear-time Algorithm of \texorpdfstring{\cite{DAESCU201922}}{}}
\label{subsec:setup}

We run the $\OO(n)$ time optimal-guard algorithm presented twice on $L$.

\smallskip
\noindent\textbf{Pass 1 (left to right).}
The algorithm scans from left to right and produces $k$ guard positions $g^+_1 < g^+_2 < \cdots < g^+_k$ on $L$ and $k$ \emph{right candidate witness points} $w^+_1 <_x w^+_2 <_x \cdots <_x w^+_k$ on $\mo$, in order of increasing $x$-coordinate. The guard $g^+_i$ is placed at $r(w^+_i, y)$, the right endpoint of the visibility interval of $w^+_i$ on $L$. The next witness point $w^+_{i+1}$ is the point on $\mo$ with the leftmost right endpoint that lies beyond $w^+_i$; it is either a vertex of $\mo$, or the point where the line through $g^+_i$ and a reflex vertex of $\mo$ meets $\partial\mo$ when extended (see~\cref{left-pass}).

\smallskip
\noindent\textbf{Pass 2 (right to left).}
The algorithm scans symmetrically from right to left, using left endpoints in place of right endpoints, and produces $k$ guard positions $g^-_1 < g^-_2 < \cdots < g^-_k$ on $L$ and $k$ \emph{left candidate witness points} $w^-_1 <_x w^-_2 <_x \cdots <_x w^-_k$ on $\mo$. Each guard $g^-_i = \ell(w^-_i, y)$ is placed at the left endpoint of the visibility interval of $w^-_i$ on $L$. The next witness $w^-_{i-1}$ is the point on $\mo$ with the rightmost left endpoint lying to the left of $w^-_i$; it is either a vertex of $\mo$, or the intersection of $\partial\mo$ with the line through $g^-_i$ and a reflex vertex of $\mo$~(see~\cref{right-pass}).

\begin{figure}[ht!]
     \centering
     \begin{subfigure}[b]{0.48\textwidth}
         \centering
         \includegraphics[width=\textwidth]{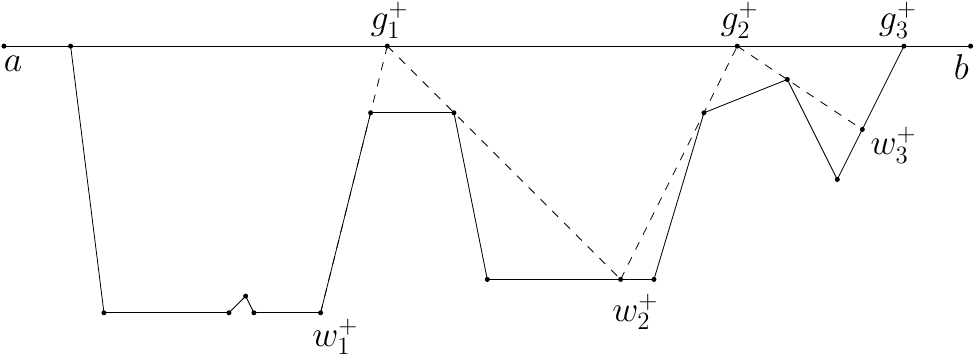}
         \subcaption{\textbf{Pass 1}: left to right}
         \label{left-pass}
     \end{subfigure}
     \hfill
     \begin{subfigure}[b]{0.48\textwidth}
         \centering
         \includegraphics[width=\textwidth]{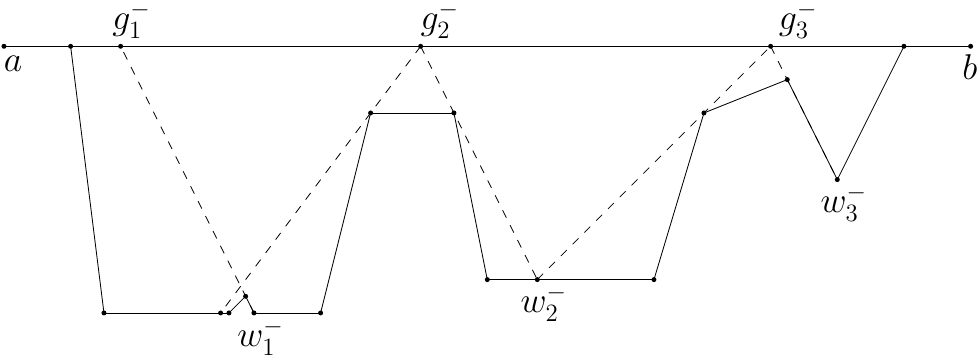}
         \subcaption{\textbf{Pass 2}: right to left}
         \label{right-pass}
     \end{subfigure}
        \caption{Two-way Linear Time Algorithm of \cite{DAESCU201922}. Pass 1 places guards $g_i^+$ at the right interval endpoints of witnesses $w_i^+$; Pass 2 symmetrically places guards $g_i^-$ at the left endpoints of witnesses $w_i^-$.} \label{fig-two-way}
\end{figure}

The two passes bracket every true witness between a left and a right candidate, in a strictly interleaved order (see \cref{fig:two-pass}).

\begin{lemma}[Interleaving of candidate witness points]\label{lem:interleave}
The $2k$ candidate witness points satisfy:
\[
  w^-_1 <_x\; w^+_1 <_x\; w^-_2 <_x\; w^+_2 <_x\; \cdots\; <_x\; w^-_k <_x\; w^+_k,
\]
and each true witness $w_i$ of $\mo$ on $L$ satisfies $w^-_i \le_x w_i \le_x w^+_i$.
\end{lemma}

\begin{figure}[H]
    \centering
    \includegraphics[width=1.0\linewidth]{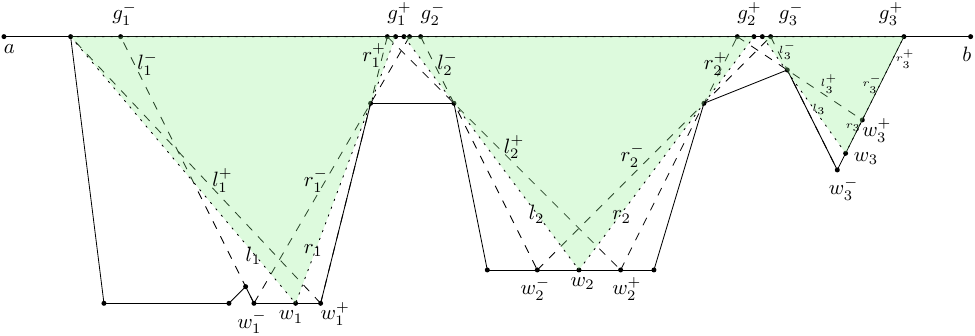}
    \caption{Interleaving of witness points and interval sandwich. Each true witness $w_i$ is bracketed by candidates $w_i^-, w_i^+$ from the two passes, with guards $g_i^-, g_i^+$ and rays $\ell_i^{\pm}, r_i^{\pm}$ sandwiching its visibility interval.}
    \label{fig:two-pass}
\end{figure}

\begin{proof}
\textit{Sandwich bound.} In Pass~1 the guard $g^+_{i-1}$ is placed at the right endpoint $r(w^+_{i-1}, y)$ of the current critical point, and $w^+_i$ is by definition the first terrain point beyond $w^+_{i-1}$ that $g^+_{i-1}$ fails to see. The true witness $w_i$ must lie at or beyond that point to avoid being covered by $g^+_{i-1}$; placing it any earlier would only shrink the region left for later witnesses without increasing their number. Hence $w_i \le_x w^+_i$, with equality when $w_i$ coincides with that point. Symmetrically, Pass~2 gives $w^-_i \le_x w_i$. Thus $w^-_i \le_x w_i \le_x w^+_i$ for every $i = 1, \dots, k$.

\textit{Interleaving.} Pass~1 partitions $\partial\mo$ into $k$ contiguous guard regions, with $w^+_i$ and $w^+_{i+1}$ marking the right boundaries of regions $i$ and $i+1$; these regions are disjoint and ordered left to right, so $w_i \le_x w^+_i <_x w^+_{i+1}$. It remains to place $w^-_{i+1}$ relative to $w^+_i$. Suppose toward a contradiction that $w^-_{i+1} \le_x w^+_i$. Since $w^-_{i+1} \le_x w_{i+1}$, the guard $g^-_{i+1} = \ell(w^-_{i+1}, y)$ sees at least as far left as $w^-_{i+1}$, and by monotonicity of visibility along the $x$-monotone terrain the guard $g^+_i$ covering the region up to $w^+_i$, together with $g^-_{i+1}$'s leftward reach, leaves no terrain strictly between them that requires a separate witness. This contradicts the existence of $k$ pairwise non-co-visible witnesses on $L$, which holds because $\mo$ is perfect \cite{DAESCU201922} and $k$ guards are optimal on $L$: $w_i$ and $w_{i+1}$ are two such witnesses, so no single guard sees both. Hence $w^+_i <_x w^-_{i+1}$, and combining $w^-_i \le_x w_i \le_x w^+_i <_x w^-_{i+1}$ for each $i = 1, \dots, k-1$ gives the full interleaving. \qedhere
\end{proof}

\subsection{Refining Candidates into True Witnesses} \label{subsec:refine}

Pass~1 fixes the guards, but its candidates $w^+_i$ are not themselves a witness set: consecutive ones may be co-visible. \textsc{RefineWitnesses} (\cref{alg:refine}) walks the sandwich of \cref{lem:interleave} once, nudging each candidate to a point that the previous guard can no longer see. It uses the following notion. For a terrain point $v \in \partial\mo$ and a guard $g \in L$ with $g \in \I_y(v)$, the \emph{shadow point} of $g$ past $v$ is
\[
  \shpt_y(g, v) \;\coloneqq\; \inf\{\, u \in \partial\mo : u >_x v,\ g \notin \I_y(u) \,\},
\]
the first terrain point beyond $v$ that $g$ fails to see on $L$.

\begin{algorithm}[H]
\caption{\textsc{RefineWitnesses}: true witnesses from candidates on a fixed line}
\label{alg:refine}
\DontPrintSemicolon
\SetKwInOut{Input}{Input}
\SetKwInOut{Output}{Output}
\SetKw{KwTo}{to}
\SetKwFunction{ShadowPoint}{shadowPoint}
\Input{Candidates $\mathcal C = \{w_j^-, w_j^+\}_{j=1}^{k}$ on $L$; height $y$.}
\Output{True witnesses $w_1, \dots, w_k$ on $\partial\mo$.}
\BlankLine
$w_1 \leftarrow w_1^+$;\quad $g_1 \leftarrow r(w_1, y)$\;
\For{$j \leftarrow 2$ \KwTo $k$}{
    $p_{w_j} \leftarrow \ShadowPoint_y(g_{j-1}, w_{j-1})$
        \tcp*{\textcolor{blue}{amortized $\OO(1)$, \cref{lem:refine-correct}}}
    \eIf{$\ell(w_j^+, y) >_x g_{j-1}$}{
        $w_j \leftarrow w_j^+$ \tcp*{\textcolor{blue}{easy case}}
    }{
        $w_j \leftarrow$ midpoint of $[w_j^-,\, p_{w_j}]$ on $\partial\mo$
            \tcp*{\textcolor{blue}{well defined by \cref{lem:sandwich-exist}}}
    }
    $g_j \leftarrow r(w_j, y)$\;
}
\KwRet{$\{w_1,\dots,w_k\}$}\;
\end{algorithm}

\begin{lemma}[Sandwich existence] \label{lem:sandwich-exist}
Suppose $L$ carries $k$ guards, and $k$ is the minimum number of guards required on $L$. Then \textsc{RefineWitnesses} is well defined at every step: whenever the otherwise-branch fires at index $j$,
\[
  w_j^- \;<_x\; p_{w_j} \;<_x\; w_j^+,
\]
so the midpoint defining $w_j$ exists and lies strictly between $w_j^-$ and $w_j^+$. Consequently $w_j^- \le_x w_j \le_x w_j^+$ for every $j = 1, \dots, k$.
\end{lemma}

\begin{proof}
We induct on $j$, maintaining the invariant $(\star)$: $w_{j-1}^- \le_x w_{j-1} \le_x w_{j-1}^+$ and $g_{j-1} \le_x g_{j-1}^+$, where $g_{j-1}^+ = r(w_{j-1}^+, y)$. The base case $j = 1$ holds with equality, since $w_1 = w_1^+$.

\smallskip
\noindent\textit{Upper bound $p_{w_j} \le_x w_j^+$.}
The shadow-point map is non-decreasing in its guard argument: on an $x$-monotone terrain, moving a guard rightward along $L$ can only move the first point it fails to see further right, since it cannot lose visibility of a point without first losing visibility of every terrain point nearer to its own witness on the same side, the mechanism underlying the recurrence in \cite{DAESCU201922}. By the inductive hypothesis $g_{j-1} \le_x g_{j-1}^+$,
\[
  p_{w_j} = \shpt_y(g_{j-1}, w_{j-1})
    \;\le_x\; \shpt_y(g_{j-1}^+, w_{j-1}^+)
    \;\le_x\; w_j^+.
\]
The last inequality holds because $w_j^+$ is, by construction of Pass~1, the leftmost point whose right endpoint could be adopted as the next witness; it dominates the shadow point of $g_{j-1}^+$ whether $w_j^+$ arises as that shadow point or as a terrain vertex overriding it.

\smallskip
\noindent\textit{Lower bound $p_{w_j} \ge_x w_j^-$.}
Suppose toward a contradiction that $p_{w_j} <_x w_j^-$. Then $g_{j-1}$ sees every terrain point strictly between $w_{j-1}$ and $w_j^-$, including points arbitrarily close to $w_j^-$. But $w_j^-$ is precisely the point whose appearance as a separate witness in Pass~2 certifies that no single guard covering the region up to $w_{j-1}$'s neighborhood can also reach $w_j^-$. If $g_{j-1}$ already reaches past $w_j^-$, then $g_1, \dots, g_{j-1}$ together with a continuation of Pass~1 to the right of $g_{j-1}$ cover all of $\partial\mo$ using at most $k-1$ guards on $L$, contradicting the minimality of $k$. Hence $p_{w_j} \ge_x w_j^-$.

\smallskip
\noindent\textit{Strictness.} 
If $p_{w_j} = w_j^+$, the easy-case test $\ell(w_j^+, y) >_x g_{j-1}$ would already have fired, so reaching the otherwise-branch forces $p_{w_j} <_x w_j^+$. Equality $p_{w_j} = w_j^-$ would require the shadow point of $g_{j-1}$, the point where the visibility ray from $g_{j-1}$ grazing the relevant reflex vertex meets $\partial\mo$, to coincide with $w_j^-$, which is determined by a generically different visibility ray or a distinct terrain vertex; the standing general-position assumption, that no two of the $\OO(k^2)$ relevant ray/terrain incidences coincide, excludes this.

Thus $w_j^- <_x p_{w_j} <_x w_j^+$ whenever the otherwise-branch fires, so the midpoint is well defined and satisfies $w_j^- <_x w_j <_x w_j^+$; together with the easy case ($w_j = w_j^+$), this restores $(\star)$ at index $j$.
\end{proof}

\begin{lemma}[Correctness, validity, and cost of \textsc{RefineWitnesses}] \label{lem:refine-correct}
On the $2k$ candidates produced by the two passes on a line $L$ carrying $k$ guards, with $k$ optimal on $L$, \textsc{RefineWitnesses} returns $k$ points $w_1, \dots, w_k$ on $\partial\mo$ with $w_j^- \le_x w_j \le_x w_j^+$ for every $j$, whose visibility intervals on $L$ are pairwise disjoint; hence $\{w_1, \dots, w_k\}$ is a witness set of size $k$ on $L$. The procedure runs in $\OO(n)$ time, with each call to $\shpt$ costing amortized $\OO(1)$.
\end{lemma}

\begin{proof}
The sandwich bounds are \cref{lem:sandwich-exist}. For disjointness, each $w_j$ is placed at or beyond $\shpt_y(g_{j-1}, w_{j-1})$: in the easy case since $\ell(w_j^+, y) >_x g_{j-1}$ directly, and in the midpoint case since the midpoint of $[w_j^-, p_{w_j}]$ lies at or beyond $p_{w_j}$, which $g_{j-1}$ fails to see. In either case $\ell(w_j, y) >_x g_{j-1} = r(w_{j-1}, y)$, so the consecutive intervals $\I_y(w_{j-1})$ and $\I_y(w_j)$ are disjoint. Since $w_1 <_x \cdots <_x w_k$ and visibility intervals vary monotonically along the $x$-monotone terrain, disjointness of every consecutive pair yields pairwise disjointness of all $k$ intervals.

For the running time, each call $\shpt_y(g_{j-1}, w_{j-1})$ walks outward along $\partial\mo$ from $w_{j-1}$, advancing a scan pointer until visibility from $g_{j-1}$ fails, exactly the mechanism presented in \cite{DAESCU201922} for locating a shadow point. By $x$-monotonicity and the strict order $w_1 <_x \cdots <_x w_k$, this pointer only advances and never backtracks across the $k$ calls, so every vertex of $\partial\mo$ is charged to at most one shadow-point computation, giving $\OO(n)$ total by the two-pointer argument, that is, amortized $\OO(1)$ per candidate.
\end{proof}




Combining the two passes with the refinement gives this result.
Running the two passes of the linear time algorithm presented in \cite{DAESCU201922} followed by \textsc{RefineWitnesses} produces the guard set $\{g_1^+, \dots, g_k^+\}$ and witness set $\{w_1, \dots, w_k\}$ as follows:
\begin{enumerate}
  \item The Pass-1 guards $\{g_1^+, \dots, g_k^+\}$ form an optimal guard set of size $k$ \cite{DAESCU201922}.
  \item By \cref{lem:refine-correct}, \textsc{RefineWitnesses} returns a witness set $\{w_1, \dots, w_k\}$ of size $k$ on $L$, which lower-bounds the guard number.
  \item Hence, the guard number equals the witness number, both equal to $k$, so the pair certifies optimality (perfectness) of the guard set.
  \item The two passes cost $\OO(n)$ and \textsc{RefineWitnesses} costs $\OO(n)$ (\cref{lem:refine-correct}), for $\OO(n)$ total, improving the $\OO(n^2\log n)$ bound of \cite{DAESCU201922} for producing a certifying witness set.
\end{enumerate}
Thus, we have proved the main result of this section.
\fhpg*


\section{Guarding a \texorpdfstring{$k$}{k}-guardable Monotone Mountain from the Lowest Horizontal Line} \label{sec:corollarykguard}

We now let the altitude line move. Given a $k$-guardable $x$-monotone mountain $\mo$ and an integer $1 \le m \le k$, we compute the minimum height $y^\ast$ at which $m$ guards on the altitude line $L_{y^\ast}$ suffice to guard all of $\mo$, together with the guard positions and $m$ witnesses certifying the count. The procedure runs in $\OO(nk + k^2\log k)$ time and is elementary, improving the $\OO(k^2\lambda_{k-1}(n)\log n)$ bound of Kang, Kim, and Ahn \cite{DBLP:conf/iwoca/KangKA25}, where $\lambda_{k-1}(n)$ is a near-linear Davenport--Schinzel length; our sweep uses no such sequences.

The engine is the fixed-height routine of \cref{sec:fixed-height}. We package the two passes and the refinement as a single subroutine.

\smallskip
\noindent\textbf{\textsc{DaescuTwoPass}$(\mo, L_y)$.}
Run both passes of Daescu's algorithm on $L_y$ to obtain the guards $G_i = \{g_1^+, \dots, g_i^+\}$ and the candidates $\mathcal C_i = \{w_j^-, w_j^+\}_{j=1}^i$, then call $\textsc{RefineWitnesses}(\mathcal C_i, y)$ to obtain the true witnesses $W_i$. Return $(G_i, \mathcal C_i, W_i)$. By \cref{thm:fixed-height}, this costs $\OO(n)$ and, when $i$ is the optimum on $L_y$, yields $i$ optimal guards and a matching size-$i$ witness set.

As the line rises, the visibility interval of a fixed terrain point can only widen, so disjoint intervals eventually overlap, and the guard count drops. To locate the exact heights at which this happens, we track each candidate's interval endpoints as functions of $y$ along fixed lines.

\subsection{Visibility Lines and the Interval Sandwich} \label{subsec:sandwich}

For each of the $2k$ candidate witness points, we construct two lines, giving $4k$ lines in total. By $x$-monotonicity of $\mo$, no two vertices share an $x$-coordinate; in particular, no witness point and its reflex anchor are vertically aligned, so each of the $4k$ lines is non-vertical.

\begin{definition}[Candidate visibility lines] \label{def:candidate-lines}
For $i = 1, \ldots, k$, let $\LA^-_i$ and $\RA^-_i$ (respectively $\LA^+_i$ and $\RA^+_i$) denote the left and right anchors of $w^-_i$ (respectively $w^+_i$), the first reflex vertices on the shortest paths from that witness toward the left and right base endpoints. Define four lines per index $i$:
\begin{alignat*}{2}
  \ell^-_i &: \text{ line through } w^-_i \text{ and } \LA^-_i,
  &\qquad
  r^-_i    &: \text{ line through } w^-_i \text{ and } \RA^-_i, \\
  \ell^+_i &: \text{ line through } w^+_i \text{ and } \LA^+_i,
  &\qquad
  r^+_i    &: \text{ line through } w^+_i \text{ and } \RA^+_i.
\end{alignat*}
At height $y$, write $\ell^-_i(y)$, $r^-_i(y)$, $\ell^+_i(y)$, $r^+_i(y)$ for the $x$-coordinates of their intersections with $L_y$.
\end{definition}

\noindent
Let $\ell_i(y)$ and $r_i(y)$ denote the left and right endpoints of the true visibility interval $\I_y(w_i)$ on $L_y$. The four lines of index $i$ bracket these true endpoints at every height.

\begin{observation}[Interval sandwich] \label{obs:sandwich}
For each $i \in \{1, \ldots, k\}$ and every height $y \ge y_{\min}$,
\[
  \ell^+_i(y) \;\le\; \ell_i(y) \;\le\; \ell^-_i(y)
  \qquad\text{and}\qquad
  r^+_i(y) \;\le\; r_i(y) \;\le\; r^-_i(y).
\]
\end{observation}

\begin{proof}
By \cref{lem:interleave} $w^-_i <_x w_i <_x w^+_i$, and for each $i \notin \{1, k\}$ the three witnesses $w^-_i$, $w_i$, $w^+_i$ share a common right anchor and a common left anchor. Denote the common right anchor by $\RA$. Since the three witnesses are ordered strictly left to right, and all three lie below $\RA$, which lies below $L_y$. Because the three rays from $w^-_i$, $w_i$, $w^+_i$ pass through the common point $\RA$, their order along $L_y$ reverses that at the witness level; hence $r^+_i(y) \le r_i(y) \le r^-_i(y)$. A symmetric argument through the shared left anchor $\LA$ gives $\ell^+_i(y) \le \ell_i(y) \le \ell^-_i(y)$.
\end{proof}

\subsection{Critical Heights: Separation and Merger} \label{subsec:critical-heights}

For consecutive witnesses $w_i$ and $w_{i+1}$, two critical heights bracket the altitude range in which their intervals pass from disjoint to overlapping.

\begin{definition}[Separation height $H_1(i,j)$ and merger height $H_2(i,j)$, $i<j$] \label{def:H1H2}
\[
  H_1(i, j) \;\coloneqq\; y\text{-coordinate of }\; r^-_i \cap \ell^+_{j},
  \qquad
  H_2(i, j) \;\coloneqq\; y\text{-coordinate of }\; r^+_i \cap \ell^-_{j}.
\]
\end{definition}

\begin{figure}[H]
    \centering
    \includegraphics[width=0.9\linewidth]{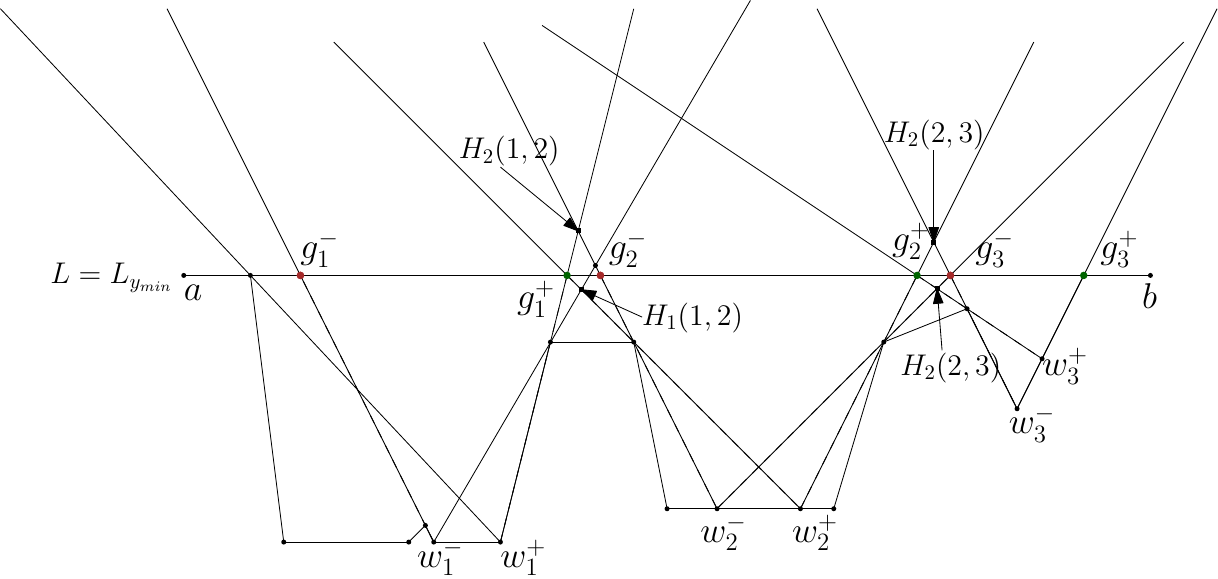}
    \caption{\cref{obs:sep-merge} illustration. The sweep line $L = L_y$ crosses the bracketing lines $\ell^{\pm}_i, r^{\pm}_i$, whose pairwise intersections mark the separation and merger heights $H_1(i,j)$ and $H_2(i,j)$.}
    \label{fig:perfect-pass}
\end{figure}

\begin{observation}[Separation and merger (\cref{fig:perfect-pass})] \label{obs:sep-merge}
For each $i, j \in \{1, \ldots, k-1\}$ with $i < j$:
\begin{enumerate}
  \item[\textup{(a)}] $H_1(i, j) \le H_2(i, j)$.
  \item[\textup{(b)}] For every $y < H_1(i, j)$, the intervals $\I_y(w_i)$ and $\I_y(w_j)$ on $L_y$ are disjoint, so two separate guards are required for $w_i$ and $w_j$.
  \item[\textup{(c)}] For every $y \ge H_2(i, j)$, the guard $g^+_i$ tracked along $r^+_i$ lies within $\I_y(w_j)$, so a single guard sees both $w_i$ and $w_j$.
\end{enumerate}
\end{observation}

\subsection{Iterative Height Reduction} \label{subsec:iterative}

\begin{algorithm}[ht]
\caption{Iterative minimum-height multi-guard placement on a $k$-guardable monotone mountain}
\label{alg:iterative-height}
\DontPrintSemicolon
\SetKwInOut{Input}{Input}
\SetKwInOut{Output}{Output}
\SetKwFunction{DaescuTwoPass}{DaescuTwoPass}
\SetKwFunction{BracketLines}{BracketLines}
\Input{A $k$-guardable $x$-monotone mountain $\mo$; an integer $1 \le m \le k$.}
\Output{Height $y^*$, guards $G^* = \{g_1^+, \dots, g_m^+\} \subseteq L_{y^*}$
        guarding $\mo$, and $m$ witnesses $W^* \subseteq \partial\mo$.}
\BlankLine
\SetKwProg{Fn}{Function}{:}{}
\Fn{\DaescuTwoPass{$\mo$, $L_y$}}{
    Run both passes of Daescu's algorithm on $L_y$ to get $G_i \leftarrow \{g_1^+,\dots,g_i^+\}$ and $\mathcal C_i \leftarrow \{w_j^-,w_j^+\}_{j=1}^i$\;
    $W_i \leftarrow \textsc{RefineWitnesses}(\mathcal C_i, y)$\;
    \KwRet{$(G_i, \mathcal C_i, W_i)$}\;
}
\BlankLine
\tcp{--- Initialization ---}
$y_0 \leftarrow y_{\min}$\;
$(G_k, \mathcal C_k, W_k) \leftarrow \DaescuTwoPass(\mo, L_{y_0})$\;
$i \leftarrow k$\;
\BlankLine
\tcp{--- Main loop ---}
\While{$i > m$}{
    \tcp{Build the $4i$ candidate visibility lines from $\mathcal C_i$}
    $\{\ell^-_j, r^-_j, \ell^+_j, r^+_j\}_{j=1}^{i} \leftarrow \BracketLines(\mathcal C_i)$\;
    \tcp{Locate the smallest consecutive merger height}
    $y_1 \leftarrow \displaystyle\min_{1 \le j \le i-1} H_2(j, j+1)$
    \tcp{Recompute exactly on the new baseline}
    $y_0 \leftarrow y_1$\;
    $(G_{i-1}, \mathcal C_{i-1}, W_{i-1}) \leftarrow \DaescuTwoPass(\mo, L_{y_0})$
        \tcp*{\textcolor{blue}{$\OO(n)$; $i-1$ guards by \cref{lem:single-step}(a)}}
    $i \leftarrow i - 1$\;
}
\BlankLine
$y^* \leftarrow y_0$;\quad $G^* \leftarrow G_m$;\quad $W^* \leftarrow W_m$\;
\KwRet{$(y^*, G^*, W^*)$}
    \tcp*{\textcolor{blue}{certifiably minimal by \cref{lem:single-step}(b)}}
\end{algorithm}

We sweep upward by a sequence of exact recomputations. At each stage, we hold two pieces of data on the current baseline: the $i$ Pass-1 guards $g^+_1, \dots, g^+_i$, and the $2i$ candidate witness points $\{w^-_j, w^+_j\}_{j=1}^i$ on $\partial\mo$ that the two passes jointly produce (\cref{subsec:setup}). \cref{alg:iterative-height} drives the sweep, calling \textsc{DaescuTwoPass} once per height.

The reported guards are exactly the Pass-1 positions $g_1^+, \dots, g_i^+$; the auxiliary guards computed inside \textsc{RefineWitnesses} only locate the witnesses and are discarded. Given $i$ guards, $2i$ candidates, and their $4i$ bracket lines on $L_{y_0}$, the loop computes the smallest consecutive merger height $y_1 = \min_{1 \le j \le i-1} H_2(j, j+1)$ in $\OO(i\log i)$ time by a single kinetic-heap scan over the two sorted line families $\{r^+_j\}$ and $\{\ell^-_j\}$, each already ordered by slope since the guards are listed in $x$-order. It then re-invokes \textsc{DaescuTwoPass} at $L_{y_1}$ in $\OO(n)$ time, obtaining fresh guards, candidates, and witnesses recomputed exactly at the new height rather than inferred from the old ones. It sets $y_0 \coloneqq y_1$, $i \coloneqq i-1$, and repeats, stopping after $k - m$ steps when $i = m$, at which point it returns $(y^\ast, G^\ast, W^\ast) \coloneqq (y_0, G_m, W_m)$.

The refinement remains valid not only at each recomputation height but also throughout the open range below the next merger, which certifies the minimality of $y^\ast$.

\begin{lemma}[Height-parametrized construction] \label{lem:height-param}
Suppose $L_{y_0}$ carries $i$ guards, $i$ is the minimum number of guards required on $L_{y_0}$, and $\{\ell_j^-, r_j^-, \ell_j^+, r_j^+\}_{j=1}^i$ are the $4i$ bracket lines built from the $2i$ candidates at $y_0$. Then for every height
\[
  y \;\in\; \Big[\, y_0,\ \min_{1 \le j \le i-1} H_2(j, j+1) \Big),
\]
running \textsc{RefineWitnesses} against the same $2i$ candidates and bracket lines is well defined at every step and produces $i$ points $w_1(y), \dots, w_i(y)$ on $\partial\mo$ whose true visibility intervals on $L_y$ are pairwise disjoint.
\end{lemma}

\begin{proof}
Write $g_j(y) \coloneqq r(w_j(y), y)$ and $g_j^+(y) \coloneqq r(w_j^+, y)$.

\smallskip
\noindent\textit{Upper bound: $p_{w_j}(y) \le_x w_j^+$, for every $j$ and $y \ge y_0$.}
By induction as in \cref{lem:sandwich-exist}: $g_1(y) = g_1^+(y)$, and if $g_{j-1}(y) \le_x g_{j-1}^+(y)$, then since $w_{j-1}(y) \le_x w_{j-1}^+$ and $r(\cdot, y)$ is monotone along the terrain, $g_{j-1}(y) = r(w_{j-1}(y), y) \le_x r(w_{j-1}^+, y) = g_{j-1}^+(y)$. Monotonicity of the shadow-point map in its guard argument then gives $p_{w_j}(y) \le_x \shpt_y(g_{j-1}^+(y), w_{j-1}^+) \le_x w_j^+$. This half never uses the height.

\smallskip
\noindent\textit{Lower bound $p_{w_j}(y) \ge_x w_j^-$, for $y < H_2(j-1, j)$.}
By \cref{def:H1H2}, $H_2(j-1, j)$ is the height at which $r_{j-1}^+$ crosses $\ell_j^-$. For $y < H_2(j-1, j)$ these lines have not yet crossed, so $g_{j-1}^+(y) = r_{j-1}^+(y) <_x \ell_j^-(y) \le_x w_j^-$. With $g_{j-1}(y) \le_x g_{j-1}^+(y)$ this gives $g_{j-1}(y) <_x w_j^-$, so $g_{j-1}(y)$ has not yet reached $w_j^-$, forcing $p_{w_j}(y) \ge_x w_j^-$.

\smallskip
\noindent\textit{Applying the hypothesis on $y$.}
Since $y < \min_{1 \le \ell \le i-1} H_2(\ell, \ell+1) \le H_2(j-1, j)$ for every $j = 2, \dots, i$, the lower bound holds at every step; strictness follows as in \cref{lem:sandwich-exist}. Hence $w_j^- <_x p_{w_j}(y) <_x w_j^+$ whenever the otherwise-branch fires, so every midpoint at height $y$ is well defined and the routine produces $i$ points $w_1(y), \dots, w_i(y)$.

\smallskip
\noindent\textit{Pairwise disjointness on $L_y$.}
Each $w_j(y)$ is placed at or beyond $\shpt_y(g_{j-1}(y), w_{j-1}(y))$: in the easy case because $\ell_j^+(y) >_x g_{j-1}(y)$ directly, and in the midpoint case because the midpoint of $[w_j^-, p_{w_j}(y)]$ lies at or beyond $p_{w_j}(y)$, which $g_{j-1}(y)$ fails to see. In either case, $\ell(w_j(y), y) >_x g_{j-1}(y) = r(w_{j-1}(y), y)$, so the consecutive intervals $\I_y(w_{j-1}(y))$ and $\I_y(w_j(y))$ are disjoint. Since $w_1(y) <_x \cdots <_x w_i(y)$ and visibility intervals vary monotonically along the terrain, disjointness of every consecutive pair yields pairwise disjointness of all $i$ intervals.
\end{proof}

One rise of the baseline to the least merger height drops the guard count by exactly one, and no smaller height achieves the drop.

\begin{lemma}[Single-step reduction] \label{lem:single-step}
Suppose $L_{y_0}$ carries $i$ guards and $2i$ candidate witnesses from two passes of Daescu's algorithm, with $i$ the minimum number of guards required on $L_{y_0}$. Let $y_1 = \min_{1 \le j \le i-1} H_2(j, j+1)$, and let $(a, a+1)$ be the minimizing pair. Then:
\begin{enumerate}
  \item[\textup{(a)}] $i-1$ is the optimum number of guards required on $L_{y_1}$.
  \item[\textup{(b)}] For every $y \in [y_0,\, H_2(a, a+1))$, at least $i$ guards are required on $L_y$.
\end{enumerate}
\end{lemma}

\begin{proof}
\textit{(a)} Let $i_1 \coloneqq |W_{i-1}|$ be the number of true witnesses returned by \textsc{DaescuTwoPass} on $L_{y_1}$. On any line, Daescu's algorithm returns a guard set whose size equals the maximum number of pairwise non-co-visible points of $\mo$ on that line; since the visibility graph of a monotone mountain restricted to a horizontal line is an interval graph, hence perfect, this equals the optimum guard number there. So $i_1$ is the optimum on $L_{y_1}$, and it suffices to show $i_1 = i-1$.

\emph{$i_1 \le i-1$:} As the base line rises from $L_{y_0}$ to $L_{y_1}$, the visibility interval of every fixed point of $\mo$ can only widen, so $i_1 \le i$. Suppose toward a contradiction $i_1 = i$. By \cref{obs:sep-merge}(c) at $y = y_1 = H_2(a, a+1)$, the point $g^+_a = r^+_a(y_1)$ lies in $\I_{y_1}(w_a) \cap \I_{y_1}(w_{a+1})$, so a single guard covers both regions certified separately by $w_a, w_{a+1}$ at $y_0$. Since Pass~1 advances to a new witness only after the current guard's interval has expired, no new witness opens where $w_{a+1}$ stood; the two regions merge. By minimality of $y_1$ and general position, every other consecutive pair has $H_2(j, j+1) > y_1$ strictly, so no other pair merges at $y_1$. Hence, exactly one merger occurs, and the greedy pass returns at most $i-1$ witnesses, contradicting $i_1 = i$. Thus $i_1 \le i-1$.

\emph{$i_1 \ge i-1$:} A single merger removes exactly one region from the greedy partition. Exactly one merger, $(a, a+1)$, occurs at $y_1$, so starting from $i$ regions at $y_0$ one merger leaves $i-1$; hence $i_1 \ge i-1$. Combining, $i_1 = i-1$, proving (a).

\medskip
\textit{(b)} Fix $y \in [y_0, H_2(a, a+1))$. Since $(a, a+1) = \arg\min_{1 \le j \le i-1} H_2(j, j+1)$, we have $H_2(a, a+1) = \min_{1 \le j \le i-1} H_2(j, j+1)$, so $y$ satisfies the hypothesis of \cref{lem:height-param}. Running the refinement at this height produces $i$ points $w_1(y), \dots, w_i(y)$ with pairwise disjoint true intervals on $L_y$, an independent set of size $i$ in the interval graph of $\mo$ restricted to $L_y$. That graph is perfect, so its minimum clique cover, the minimum guard number on $L_y$, is at least $i$. As $y$ was arbitrary in the range, at least $i$ guards are required throughout; combined with part (a), exactly $i-1$ suffice at $y = H_2(a, a+1)$ itself, identifying it as the precise height at which the guard requirement drops from $i$ to $i-1$.
\end{proof}

\subsection{Correctness and Complexity} \label{subsec:correctness}

Algorithm~\ref{alg:iterative-height} computes $y^\ast$, $G^\ast$, and $W^\ast$ as follows.

\textit{Correctness.}
\begin{enumerate}
  \item By induction on the number of completed steps $t = 0, 1, \dots, k-m$, after $t$ steps the base line $L_{y_0}$ carries exactly $i = k-t$ Pass-1 guards and $i$ true witnesses from $\textsc{DaescuTwoPass}(\mo, L_{y_0})$, with $i$ optimal on $L_{y_0}$.
  \item The base case $t=0$ holds by construction: as $\mo$ is $k$-guardable, $\textsc{DaescuTwoPass}(\mo, L_{y_{\min}})$ returns $k$ Pass-1 guards, and by perfectness of the interval graph this equals the true optimum on $L_{y_{\min}}$.
  \item For the inductive step, \cref{lem:single-step}(a) gives that a fresh call to \textsc{DaescuTwoPass} at $y_1 = \min_{1 \le j \le i-1} H_2(j,j+1)$ has optimum $i-1$ on $L_{y_1}$, restoring the hypothesis at $t+1$.
  \item After $t = k-m$ steps, $i=m$, so $G^\ast = G_m$ guards $\mo$ with $|G^\ast|=m$, and by \cref{lem:refine-correct} the witnesses $W^\ast = W_m$ are $m$ well-defined points on $\partial\mo$ with $w_j^- \le_x w_j \le_x w_j^+$, from the same final call.
\end{enumerate}

\textit{Minimality of $y^\ast$.}
\begin{enumerate}
  \item Let $(a,a+1)$ minimize $H_2(j,j+1)$ at the last step, so $y^\ast = H_2(a,a+1)$.
  \item By \cref{lem:single-step}(b) applied at that step, for every $y \in [y_0', y^\ast)$ (with $y_0'$ the baseline held before the last reduction), at least $m+1$ guards are required on $L_y$; the height-parametrized witnesses of \cref{lem:height-param} furnish an explicit independent set of size $m+1$ for each such $y$.
  \item Hence $m$ guards do not suffice strictly below $y^\ast$ in that range, and $y^\ast$ is the least height at which $m$ guards suffice.
\end{enumerate}

\textit{Complexity.}
\begin{enumerate}
  \item Step $i$ (for $i = k, k-1, \dots, m+1$) invokes \textsc{DaescuTwoPass} once, costing $\OO(n)$ for the two passes plus $\OO(n)$ for \textsc{RefineWitnesses} (\cref{lem:refine-correct}).
  \item Each step also costs $\OO(i\log i)$ for constructing the $4i$ bracket lines and locating $y_1 = \min_{1 \le j \le i-1} H_2(j,j+1)$ via the kinetic scan.
  \item Summing over the $k-m$ steps,
  \[
    \sum_{i=m+1}^{k} \big( \OO(n) + \OO(i\log i) \big)
      \;=\; O\big((k-m)\,n\big) + O\!\Big( \sum_{i=1}^{k} i\log i \Big)
      \;=\; \OO(nk + k^2\log k).
  \]
  \item This improves the $\OO(k^2\lambda_{k-1}(n)\log n)$ bound of \cite{DBLP:conf/iwoca/KangKA25} and uses no Davenport--Schinzel sequences.
\end{enumerate}
Thus, we have proved the main result of this section.
\mhg*


\section{Conclusion}

The {\sc Art Gallery Problem}, in full generality, is one of the most stubborn problems in computational geometry: it is $\exists\mathbb{R}$-complete, resistant to exact algorithms, and acutely sensitive to the coordinates of the polygon's vertices. We asked a simple question: what happens when guards cannot roam freely, but are anchored to a single edge, a shoreline, a rail, a fixed base? The answer turned out to be clean.

The first insight is that this restriction does more than simplify the problem; it changes the problem's structure. For a weak visibility polygon $\wv$ with base $\eb$, guarding the boundary $\bd(\wv)$ is the same as guarding the entire polygon. This is not obvious, since the interior of a $\wvp$ can hide deep pockets and intricate shadow regions. A simple triangle argument settles it: every interior point lies behind some already-covered boundary point, so the base guard watching that boundary point sees the interior point too. With this reduction in hand, a continuous two-dimensional coverage problem collapses to a one-dimensional boundary-covering task.

The second insight is geometric. Every boundary point $p$ carries a natural \emph{visibility interval} $\I(p)$ on $\eb$, fixed by just two reflex vertices: the first obstacles on the shortest paths from $p$ toward the two endpoints of the base. A guard on $\eb$ sees $p$ if and only if it lies in $\I(p)$, which turns a two-dimensional visibility question into a one-dimensional one: rather than reason about line-of-sight geometry in the plane, we pierce intervals on a line. These intervals form a \emph{visibility interval graph}, and interval graphs are perfect, a classical fact with a strong consequence: the largest witness set and the smallest clique cover have equal size. The guard lower bound, from witnesses, and the guard upper bound, from an actual guard set, are thus achieved together, and the problem becomes self-certifying.

For vertex guards, this gives a complete answer: an $\OO(n\log n)$ algorithm built from the linear-time shortest-path trees by the authors in \cite{GUIBAS1989126} and a minimum clique cover of the resulting interval graph, matched by an $\Omega(n\log n)$ lower bound via a reduction from Sorting. Vertex guarding is therefore settled at $\Theta(n\log n)$. For full polygon guarding, the story is richer because boundary points other than vertices can spawn witnesses that were invisible at the start. The Witness-Guard Algorithm meets this by expanding its candidate set as it runs: each time a guard is placed, rays through the reflex vertices it sees are shot to the opposite boundary, generating new candidate witnesses. Its greedy rule, pierce the uncovered interval with the leftmost right endpoint, is optimal by the perfectness of interval graphs, and the expansion rule leaves no shadow region undetected. Two theorems close the loop: Optimality shows the algorithm never places a superfluous guard, and Completeness shows it never misses a corner. Together they give an exact $\OO\bigl((n + \OPT\cdot\rho)\,(\log n + \log\OPT)\bigr)$ algorithm for a problem that is $\exists\mathbb{R}$-complete in the unconstrained setting.

As a bonus, the same machinery resolves an open question of \cite{DAESCU201922} about monotone mountains: given a $k$-guardable mountain and a target $m \le k$, find the lowest altitude at which $m$ guards on a horizontal line suffice. Kang, Kim, and Ahn \cite{DBLP:conf/iwoca/KangKA25} answered it by tracking how the visibility intervals move as the altitude rises, which led them to Davenport--Schinzel sequences and an $\OO(k^2\lambda_{k-1}(n)\log n)$ bound. We take a different route. Two passes of Daescu's algorithm bracket each true witness between a left and a right candidate, producing $4k$ lines whose crossings are exactly the heights at which one guard's coverage absorbs another's. An upward sweep locates these merger heights, and a single exact recomputation at each of the $k-m$ drops yields the guards together with their certifying witnesses in $\OO(nk + k^2\log k)$ time, with no Davenport--Schinzel machinery.

\medskip
The picture is not yet complete. Two directions seem especially worth pursuing.
\begin{itemize}
    
    \item The most immediate concern is the runtime. Can full polygon guarding be solved in $\OO(\OPT + n\log n)$ time, matching the vertex-guard bound? The present bottleneck is the per-iteration cost of sorting candidate intervals and expanding the witness set, together $\OO(n\log n)$ per guard placed. An incremental structure that updates this information as guards and witnesses arrive, rather than rebuilding it, could close the gap.

    \item The second is structural, and concerns where guards may stand. What changes if guards may lie anywhere on the boundary $\bd(\wv)$, not only on $\eb$? The interval representation breaks down because a guard on the upper chain has no clean interval on the base, and entirely new geometric ideas seem necessary.

\end{itemize}

\paragraph*{Declaration on the use of AI.}
The authors used Claude (Anthropic) to assist with drafting and revising portions of the manuscript text. All mathematical results, proofs, and technical content are the authors' own.

\printbibliography

\end{document}